\documentclass{amsart}
\usepackage{fancyhdr}
\usepackage{amsfonts}
\usepackage{graphicx}
\usepackage{epstopdf}
\usepackage{algorithmic}
\usepackage{amsmath}
\usepackage{bm}

\usepackage{comment}
\usepackage{url}
\usepackage[caption=false,font=normalsize,labelfont=sf,textfont=sf]{subfig}
\usepackage{algorithm}
\newtheorem{theorem}{Theorem}[section]
\newtheorem{lemma}[theorem]{Lemma}

\theoremstyle{definition}
\newtheorem{definition}[theorem]{Definition}
\newtheorem{example}[theorem]{Example}

\theoremstyle{remark}
\newtheorem{remark}[theorem]{Remark}

\numberwithin{equation}{section}
\newtheorem{corollary}{Corollary}

\newcommand{\beq}{\begin{equation}}
	\newcommand{\eeq}{\end{equation}}
\newcommand{\beqnn}{\begin{equation*}}
	\newcommand{\eeqnn}{\end{equation*}}
\newcommand{\beqy}{\begin{eqnarray}}
	\newcommand{\eeqy}{\end{eqnarray}}
\newcommand{\beqynn}{\begin{eqnarray*}}
	\newcommand{\eeqynn}{\end{eqnarray*}}
\newcommand{\bit}{\begin{itemize}}
	\newcommand{\eit}{\end{itemize}}
\newcommand{\ben}{\begin{enumerate}}
	\newcommand{\een}{\end{enumerate}}
\newcommand{\bex}{\begin{example}}
	\newcommand{\eex}{\end{example}}
\newcommand{\norm}[1]{\left\lVert#1\right\rVert}

\newcommand{\abs}[1]{\left| #1 \right|}

\newcommand{\A}{\boldsymbol{A}}

\newcommand{\F}{\boldsymbol{F}}
\newcommand{\G}{\boldsymbol{G}}

\newcommand{\I}{\boldsymbol{I}}
\newcommand{\J}{\boldsymbol{J}}

\renewcommand{\L}{\boldsymbol{L}}
\newcommand{\M}{\boldsymbol{M}}
\newcommand{\N}{\boldsymbol{N}}

\newcommand{\R}{\mathbb{R}}

\newcommand{\Y}{{\boldsymbol{Y}}}

\newcommand{\h}{\boldsymbol{h}}

\newcommand{\p}{\boldsymbol{p}}
\newcommand{\q}{\boldsymbol{q}}

\newcommand{\TT}{\top}

\renewcommand{\u}{\boldsymbol{u}}
\renewcommand{\v}{\boldsymbol{v}}

\newcommand{\y}{{\boldsymbol{y}}}
\newcommand{\x}{\bm{x}}
\newcommand{\z}{\bm{z}}

\newcommand{\1}{{\boldsymbol{1}}}
\newcommand{\mbbR}{{\mathbb{R}}}

\newcommand{\mbbE}{{\mathbb{E}}}
\newcommand{\mbbP}{{\mathbb{P}}}
\newcommand{\mbbS}{{\mathbb{S}}}
\newcommand{\mcalS}{\mathcal{S}}
\newcommand{\mcalT}{\mathcal{T}}
\newcommand{\mcalX}{\mathcal{X}}
\newcommand{\mcalN}{\mathcal{N}}
\newcommand{\mcalR}{\mathcal{R}}
\numberwithin{equation}{section}

\ifpdf
  \DeclareGraphicsExtensions{.eps,.pdf,.png,.jpg}
\else
  \DeclareGraphicsExtensions{.eps}
\fi

\usepackage{amsopn}

\begin{document}
\title[Sparse Signal Recovery From Quadratic Systems]{Sparse Signal Recovery From Quadratic Systems with Full-Rank Matrices}

\author{Jinming~Wen}
\address{College of Information Science and Technology, Jinan University, Guangzhou, China}

%\curraddr{Department of Mathematics and Statistics,
%Case Western Reserve University, Cleveland, Ohio 43403}

\email{jinming.wen@mail.mcgill.ca}
\thanks{The first author was supported by NSFC Grant Nos. 12271215, 12326378 and 11871248.}

\author{Yi~Hu}
\address{College of Information Science and Technology, Jinan University, Guangzhou, China}
\email{huyijnu@stu2023.jnu.edu.cn}

\author{Meng~Huang}
\address{School of Mathematical Sciences, Beihang University, Beijing, 100191, China.}

%\curraddr{Department of Mathematics and Statistics,
	%Case Western Reserve University, Cleveland, Ohio 43403}

\email{menghuang@buaa.edu.cn}
\thanks{The third author was supported by NSFC Grant No. 12201022.
}
\thanks{The third author is the corresponding author.
}
%    General info
\subjclass[2020]{Primary 65F10, 15A29, 94A12}

\date{January 1, 2001 and, in revised form, June 22, 2001.}


\keywords{Gauss-Newton method, sparse signals, the number of measurements, quadratic system.}

\begin{abstract}
In signal processing and data recovery, reconstructing a signal from quadratic measurements poses a significant challenge, particularly in high-dimensional settings where measurements $m$ is far less than the signal dimension $n$ (i.e., $m \ll n$). This paper addresses this problem by exploiting signal sparsity. Using tools from algebraic geometry, we derive theoretical recovery guarantees for sparse quadratic systems, showing  that  $m\ge 2s$ (real case) and $m\ge 4s-2$ (complex case) generic measurements suffice to  uniquely recover all $s$-sparse signals. Under a Gaussian measurement model, we propose a novel two-stage Sparse Gauss-Newton (SGN) algorithm. The first stage employs a support-restricted spectral initialization, yielding an accurate initial estimate with $m=O(s^2\log{n})$ measurements.  The second stage refines this estimate via an iterative hard-thresholding Gauss-Newton method, achieving quadratic convergence to the true signal within finitely many iterations when $m\ge O(s\log{n})$. Compared to existing second-order methods, our algorithm achieves near-optimal sampling complexity for the refinement stage without requiring resampling. Numerical experiments indicate that SGN significantly outperforms state-of-the-art algorithms in both accuracy and computational efficiency. In particular, (1) when sparsity level $s$ is high, compared with existing algorithms, SGN can achieve the same success rate with fewer measurements. (2) SGN converges with only about $1/10$ iterations of the best existing algorithm and reach lower relative error.  
\end{abstract}
\maketitle
\thispagestyle{plain}

%%\section*{This is an unnumbered first-level section head}
%%This is an example of an unnumbered first-level heading.

%% The correct journal style for \specialsection is all uppercase; a known bug
%% in amsart.cls prevents this, so input must be uppercase until it is fixed.
%\specialsection*{This is a Special Section Head}
%%\specialsection*{THIS IS A SPECIAL SECTION HEAD}
%%This is an example of a special section head%
%%%%%%%%%%%%%%%%%%%%%%%%%%%%%%%%%%%%%%%%%%%%%%%%%%%%%%%%%%%%%%%%%%%%%%%%
%\footnote{Here is an example of a footnote. Notice that this footnote
%text is running on so that it can stand as an example of how a footnote
%with separate paragraphs should be written.
%\par
%And here is the beginning of the second paragraph.}%
%%%%%%%%%%%%%%%%%%%%%%%%%%%%%%%%%%%%%%%%%%%%%%%%%%%%%%%%%%%%%%%%%%%%%%%%
.

\section{Introduction}
\subsection{Problem setup}
We consider the problem of recovering an unknown signal $\x \in \mathbb H^n$ (where $\mathbb H=\R$ or $\mathbb H=\mathbb C$) from a system of $m$ quadratic equations
\begin{equation} \label{eq:prob}
	y_i = \x^\TT\A_i\x, i = 1,...,m,
\end{equation}
where $\A_i \in \mathbb H^{n\times n}$  are the measurement matrices, $y_i \in \mathbb H$ are observations.  This quadratic inverse problem arises naturally in numerous scientific applications where physical measurements are inherently quadratic. In computational imaging, sub-wavelength microscopy with partially incoherent light requires solving (\ref{eq:prob}) to recover fine-scale features from intensity-only measurements \cite{szameit2012sparsity,shechtman2011sparsity}. Similarly, quantum state tomography \cite{gross2010quantum,basu2000uniqueness,zehni20203d} relies on this formulation, where quantum correlations are captured as quadratic forms (\ref{eq:prob}) of the quantum state vector $\x$. The same mathematical framework appears in structural biology, particularly in cryo-electron microscopy \cite{bendory2024sample, bendory2023autocorrelation}, which  leverages quadratic measurements (\ref{eq:prob}) to reconstruct 3D molecular structures from noisy 2D projections.

Existing algorithms for this problem primarily focus on the low-dimensional re-gime where $m>n$, typically requiring $O(n)$ measurements \cite{chen2022error, huang2019solving, huang2020solving}. While effective for small-scale problems, these methods encounter computational and storage limitations in high-dimensional settings where both $m$ and $n$ are large. A key breakthrough came with the realization that many real-world signals possess inherent sparsity. Leveraging this structural prior has led to dramatic improvements in both sample complexity and computational efficiency.
For instance, Shechtman et al. \cite{shechtman2011sparsity} demonstrated that sparsity priors enable exact recovery of sub-wavelength optical images from significantly fewer measurements. Similarly, Bendory et al. \cite{bendory2023autocorrelation} proved that molecular structures can be uniquely identified from their second-order autocorrelations, dramatically reducing measurement requirements compared to non-sparse cases.

In this paper, we focus on sparse quadratic systems and aim to address two fundamental questions: 1) How many measurements are required to ensure exact recovery of all  $s$-sparse signals $\x$ from \eqref{eq:prob}? Here, $s$-sparse means $\x$ has at most $s$ non-zero elements. 2) Can we design a computationally efficient algorithm with provable quadratic convergence guarantee under Gaussian measurement models?

\subsection{Related work and contributions}
When the measurement matrices are rank-1 and positive semidefinite, the recovery problem in \eqref{eq:prob} reduces to the well-studied phase retrieval problem. Extensive research over the past two decades has produced numerous provably algorithms for this problem, including both convex approaches \cite{phaselift,Phaseliftn,Waldspurger2015} and non-convex methods \cite{AltMin,candes2015phase,TWF,turstregion,macong,TAF,RWF,Duchi,tan2019phase,huangsiam,ZYL}, all of which guarantee global convergence under $O(n)$ Gaussian random measurements. The sparse variant of this problem, where the signals $\x$ is $s$-sparse, has attracted particular attention in high-dimensional settings. Theoretical studies have established that $2s$ measurements suffice for unique recovery in the real case, while $m>4s-2$ measurements are needed in the complex case \cite{akccakaya2015sparse, wang2014phase}.  Practical algorithms such as thresholding Wirtinger flow, Sparse Truncated Amplitude flow (SPARTA), and Compressive Phase Retrieval with Alternating Minimization (CoPRAM) \cite{cai2016optimal, jagatap2019sample,akccakaya2015sparse} have been developed, achieving successful recovery with $O(s^2\log{n})$ Gaussian random measurements. While this sample complexity still shows a gap compared to theoretical limits, it represents significant computational savings compared to non-sparse cases.
For a comprehensive overview of recent theoretical, algorithmic, and application developments in phase retrieval, readers are referred to survey papers \cite{jaganathan2016phase,shechtman2015phase}.

This paper focuses on recovering $s$-sparse $\x$ from \eqref{eq:prob} with full-rank matrices $\A_i, i=1,\ldots,m$.
In low-dimensional settings ($m>n$) for non-sparse $\x$, uniqueness guarantees are well understood:  $2n-1$ and  $4n-4$ generic Hermitian matrices $\A_i, i=1,\ldots,m$, suffice for exact recovery in the real case and complex case, respectively \cite{wang2019generalized}.  However,  in the high dimensional scenario ($m<n$) where the signals are sparse, how many sampling matrices are needed to guarantee the unique recovery remains highly incomplete. Algorithmically, Huang et al. generalized the Wirtinger Flow (WF) algorithm \cite{candes2015phase} to full-rank measurements, achieving linear convergence for non-sparse recovery with $O(n)$ Gaussian measurements \cite{huang2019solving, huang2020solving}.  For sparse signals,  
Chen et al. adapted Thresholded Wirtinger Flow (TWF) \cite{cai2016optimal} to this setting, preserving linear convergence under
$O(s^2\log{n})$ Gaussian measurements.  Notably, all existing sparse recovery algorithms face a fundamental  $O(s^2 \log n)$ sample complexity bottleneck due to spectral initialization limitations.

Our contributions are summarized as:
\begin{itemize}
	\item[(i)] Leveraging tools from algebraic geometry, we establish that $m\ge 2s$ generic measurement matrices $\A_i \in \mathbb R^{n\times n}, i=1,\ldots,m$, ensure the unique recovery of $s$-sparse signals $\x \in \R^n$ from quadratic measurements \eqref{eq:prob}.  Similarly, in the complex case, $m\ge 4s-2$ generic Hermitian  matrices $\A_i \in \mathcal H^{n\times n}, i=1,\ldots,m$  suffice. \footnote{In the complex case, we restrict our analysis to Hermitian matrices due to their greater practical relevance compared to non-Hermitian ones}
	\item[(ii)] We propose the Sparse Gauss-Newton (SGN) method, an efficient second-order algorithm for solving sparse quadratic systems under Gaussian measurements. The algorithm proceeds in two key stages: (1) spectral initialization requiring $O(s^2 \log n)$ measurements to obtain a good initial estimate near the true signal, followed by (2) Gauss-Newton refinement that achieves quadratic convergence within $O(s \log n)$ measurements. The SGN offers the advantage of optimal sampling complexity during refinement without needing resampling, whereas existing methods can only achieve $O(s^2\log n)$ sample complexity without resampling. 
	Moreover, SGN demonstrates favorable performance in terms of accuracy and convergence speed.  
\end{itemize}

\subsection{Notation}
Throughout the paper, we use lowercase letters to denote vectors, and uppercase letters to denote matrices.
We use $C, c, C_1, c_1, C_2, c_2,\ldots$ to denote positive constants whose values vary with the context. 
The symbol $\abs{\cdot}$ represents either the absolute value of a scalar or the cardinality of a set. 
The notation $\norm{\cdot}$ denotes the $\ell_2$-norm of a vector or the spectral norm of a matrix, and $\norm{\cdot}_0$ represents the number of nonzero elements of a vector.
For a vector $\x$, $\text{supp}(\x)$ represents its support set. Subscripts are used to index specific elements or columns: for example, $\x_\mcalS$ refers to the vector formed by retaining the elements of $\x$ indexed by $\mcalS$ and setting the remaining elements to zero, while $\A_\mcalS$ refers to the matrix formed by retaining the columns of $\A$ indexed by $\mcalS$ and setting the remaining columns to zero. 
The notation $\A_i$ represents the $i$-th measurement matrix, and $a^i_{kl}$ denotes the element located at the $(k, l)$ position of $\A_i$.  
The notation $\A_i \sim \mcalN^{n \times n}(0,1)$ denotes that $\A_i$ is an $n \times n$ matrix, whose elements are independently and identically distributed standard normal random variables.
For a symmetric matrix $\M \in \mathbb R^{n \times n}$, we use $\lambda_{\min}$ and $\lambda_{\max}$ to denote its minimum and maximum eigenvalues, respectively.
Let $\mathbb S^{n-1}=\{\z \in \mathbb R^n : \norm{\z} = 1\}$ be the standard Euclidean sphere. For a random variable $X$, we define 
$
\norm{X}_{\psi_2}=\inf\{t>0: \mbbE \exp{(X^2/t^2)} \le 2\}
$
as the sub-gaussian norm of $X$, and 
$
\norm{X}_{\psi_1}=\inf\{t>0: \mbbE \exp{(\abs{X}/t)} \le 2\}
$
as the sub-exponential norm of $X$. If $X$ has finite $\norm{X}_{\psi_2}$, then we say $X$ is sub-gaussian. Similarly, $X$ with finite $\norm{X}_{\psi_1}$ is said to be sub-exponential.

\subsection{Organization}
The remainder of this paper is organized as follows. Sec. \ref{sec:2} establishes the fundamental sample complexity requirements for unique recovery in sparse quadratic systems. 
Sec. \ref{sec:3} introduces our proposed algorithm with complete implementation details, while Sec. \ref{sec:4} provides a rigorous convergence analysis. 
Sec. \ref{sec:5} presents comprehensive numerical experiments to illustrate the algorithm's empirical performance. We conclude with a discussion of potential extensions and future research directions in Sec. \ref{sec:6}.

\section{Sample number for sparse quadratic systems}
\label{sec:2}
In this section, we investigate the number of measurements required for recovering all $s$-sparse signals from observations \eqref{eq:prob}.  Note that $(c\x)^{\TT}\A_i(c\x)=\x^{\TT}\A_i\x$ for any $|c|=1$. Thus, one can only recover $\x$ up to a global phase. Consider the equivalence relation $\sim$ on ${\mathbb H}^n$:  $\x\sim\z$ if and only if there is a constant $c\in {\mathbb H}$ with $|c|=1$ such that $\x=c\z$. Denote $\widetilde{{\mathbb H}}^n:={\mathbb H}^n/\sim$,  and let $\widetilde{\x}:=\x/\sim$ be the equivalent class containing $\x$. For given measurement matrices $\mathcal A:=\{A_i\}_{i=1}^m$, we define the map $\mathbf{M}_{\mathcal A} : \widetilde{{\mathbb H}}^n \to \mathbb R^m$ by
\[
\mathbf{M}_{\mathcal A} (\x)=\left(\x^{\TT}\A_1\x, \ldots, \x^{\TT}\A_m\x  \right).
\]
Set 
\[
\Sigma_s ({\mathbb H}^n):=\left\{\x \in{\mathbb H}^n : \norm{\x}_0 \le s   \right\}.
\]
Before presenting the theoretical results of this paper, we first introduce some basic notations of algebraic geometry.

\subsection{Terminology from algebraic geometry}
Let $\mathbb H$ be a field (either $\mathbb H=\R$ or $\mathbb H=\mathbb C$).  An {\em affine variety} in $\mathbb H^n$ is a subset defined by the vanishing of finitely many polynomials in $\mathbb H[x_1,\cdots,x_n]$. If these polynomials are homogeneous, their vanishing defines a subset of projective space $\mathbb P(\mathbb H^n)$, which is called a {\em projective variety}. The Zariski topology on $\mathbb H^n$ is defined by taking  affine varieties as closed sets, making their complement  Zariski open sets. We say a property holds for a {\em generic point} in  $\mathbb H^n$ if there exists a non-empty Zariski open set of points satisfying it.  For any vector space $X$, we denote by $\mathbb P(X)$ the associated projective space, i.e. the set of all one-dimensional subspaces of $X$.

%We first introduce some basic concepts in algebraic geometry that are useful for this paper.
%As usual for each $\x \in X$ we use $[\x]$ to denote the induced elements in $\mathbb P(X)$.
%Similarly, for any subset $S \subset X$ we use $[S]$ or $\mathbb P(S)$ to denote its projectivization in $\mathbb P(X)$. Throughout this paper, we say $V \subset \mathbb C^n$ is a projective variety if $V$ is the locus of a collection of homogeneous polynomials in $\mathbb C[\x]$.
%Strictly speaking a projective variety lies in $\mathbb P(\mathbb C^n)$ and is the projectivization of the zero locus of a collection of homogeneous polynomials.
%But like in \cite{wang2019generalized}, when there is no confusion the phrase \textit{projective variety} in $\mathbb C^n$ means an
%algebraic variety in $\mathbb C^n$ defined by homogeneous polynomials.
%We shall use \textit{projective variety} in $\mathbb P(\mathbb C^n)$ to describe a true projective variety.

\subsection{Main results}
In this subsection, we establish the number of measurements required to recover $s$-sparse signals from quadratic measurements in both real-valued and complex-valued cases.
Our first theorem shows that $m\ge 2s$ generically chosen matrices $\A_i$ in $\mathbb R^{n\times n}$ suffice to recover all $s$-sparse signals in $\mathbb R^n$. 
\begin{theorem}
	\label{Theorem4}
	Let $m\ge 2s$ and $\mathcal A:=\{\A_i\}_{i=1}^m$ be a generic point in $\mathbb R^{m\times n^2}$. Then $\mathbf{M}_{\mathcal A}$ is injective on $\Sigma_s ({\mathbb R}^n)$.
\end{theorem}

Therefore, we establish that $m \ge 2s$ measurements are sufficient to exactly recover $\x$ based on $y_i$ and $\A_i$ for $1 \le i \le m$ (see \eqref{eq:prob}) in the real case.

Then we turn to the complex case, where, unlike the real setting, we restrict our analysis to Hermitian measurement matrices $\A_1,\ldots,\A_m \in \mathbb C^{n\times n}$, as non-Hermitian measurements are of litter practical interest. Let $\mathcal H^{n\times n}$ denote the set of Hermitian matrices over $\mathbb C^{n\times n}$. The second theorem  shows that generically chosen Hermitian matrices $\A_i$ in $\mathbb C^{n\times n}$ suffice to  guarantee  the recovery of all $s$-sparse signals in $\mathbb C^n$ if $m\ge 4s-2$.

\begin{theorem}
	\label{Theorem5}
	Let $m\ge 4s-2$. Then $\mathbf{M}_{\mathcal A}$ with a generic $\mathcal A:=\{\A_i\}_{i=1}^m \in \mathcal H^{n\times n} \times \cdots \times \mathcal H^{n\times n}$ is injective on $\Sigma_s ({\mathbb C}^n)$.
\end{theorem}

By means of this theorem, we obtain that in the complex case, $m \ge 4s-2$ measurements can guarantee the unique $s$-sparse solution for problem (\ref{eq:prob}). 

The detailed proofs of the above two theorems can be found in Appendix \ref{sec: realmeas} and \ref{sec: complexmeas}, respectively.

\section{Algorithms} \label{sec:3}
For convenience, we simply discuss the case where $\x$ and $\A_i$ are real-valued.  A common way to recover $s$-sparse signals $\x \in \R^n$ from the system of quadratic equations (\ref{eq:prob}) is to solve the following least squares problem:
\begin{equation}
	\min\limits_{\z \in \mbbR^n, \norm{\z}_0 \le s} f(\z) := \frac{1}{2m}\sum\limits_{i=1}^{m}(\z^{\TT}\A_i\z-y_i)^2,
	\label{least-square problem}
\end{equation}
where $y_i = \x^{\TT}\A_i\x$ for $i=1,\ldots, m$. We assume that  $\A_i \in \R^{n\times n}, i=1,\ldots, m$ are independently and identically distributed (i.i.d.) Gaussian random measurement matrices with standard normal entries.

The non-convexity of the loss function in (\ref{least-square problem}) renders direct optimization methods prone to convergence to a local minimizer. Following the methodology in \cite{wang2017sparse,candes2015phase,jagatap2019sample}, we adopt a two-stage algorithmic approach: (1) A spectral initialization stage that generates a high-quality initial estimate $\x^0$, and (2) an iterative refinement stage that progressively converges to the true solution through successive updates starting from $\x^0$.

\begin{algorithm}[H]
	\caption{Initialization}
	\label{algorithm1:Initialization}
	\renewcommand{\algorithmicrequire}{\textbf{Input:}}
	\renewcommand{\algorithmicensure}{\textbf{Output:}}
	\begin{algorithmic}[1]
		\REQUIRE Gaussian random matrix $\A_{i}, i=1,...,m$, measurement vector $\y \in \mbbR^m$, \ sparsity level $s$.
		\STATE Compute signal marginals: \ $ Y_{jj} = \frac{1}{m}\sum\limits_{i=1}^{m}y_ia_{jj}^i $ for $j = 1,\ldots ,n$.
		\STATE Set $ \widehat{\mcalS}$ to the set whose elements corresponding to the $s$-largest instances of  $Y_{jj}$.
		\STATE Let $\v$ be the  leading left singular vector of $\Y_{\hat{\mcalS}}=\frac{1}{m}\sum\limits_{i=1}^{m}y_i[\A_{i}]_{\hat{\mcalS},\hat{\mcalS}}$.
		\STATE Set $\phi = \left(\frac{1}{2m}\sum\limits_{i=1}^{m}y_i^2\right)^{1/4}$.
		\STATE Initial estimate $\x^0=\phi \v$.
		\ENSURE $\x^0$
	\end{algorithmic}
\end{algorithm}
\subsection{Initialization}
Before introducing our proposed initialization algorithm, we first provide an overview of existing spectral initialization methods.
Define 
\[
\Y=\frac{1}{m}\sum_{i=1}^{m}y_i\A_i\qquad \mbox{and}\qquad \phi = \left(\frac{1}{2m}\sum\limits_{i=1}^{m}y_i^2\right)^{1/4},
\]
where $y_i = \x^{\TT}\A_i\x, i=1,\ldots,m$ and $\x \in \R^n$ is the true signal. Since we assume $a^i_{kl}\sim \mathcal N(0,1)$ for $i = 1,\ldots,m$, $k= 1,\ldots,n$, $l = 1,\ldots,n$, it is straightforward to verify that the expectation of $\Y$ is $ {\mathbb E} \Y= \x\x^{\TT}$ and $ {\mathbb E} \phi=\norm{\x}$. Consequently, by the law of large number, when $m$ is large,  the leading left or right singular vector of $\Y$, scaled by $\phi$, serves as a good estimator of $\x$.  Building on this idea, \cite{huang2019solving} established that  $m=\Omega(n)$ samples are sufficient to ensure a reliable initial estimate.

When $\x$ is a sparse signal, one can leverage the sparsity prior to reduce the sample complexity significantly. The main idea is first estimate the support of $\x$, and then apply the spectral initialization restricted to the estimated support set. Specifically, denote $\mcalS=\mbox{supp}(\x)$. Then the expectation of diagonal entries of $\Y$ obey
\begin{equation*}
	{\mbbE[Y_{jj}]} = \begin{cases}
		x_j^2,&{\text{if}}\ j \in \mcalS \\ 
		{0,}&{\text{if}}\ j \in \mcalS^c
	\end{cases}.
\end{equation*}
Thus,  the magnitudes of the diagonal entries of $\Y$ can be used to estimate $\mcalS$. Based on the thresholding operator, \cite{chen2025solving} proposed a variant of the spectral initialization method for sparse signals. To be specific,  the estimated support set is given by
\begin{equation*}
	\widehat{\mcalS}_0 = \left\{j \in [n]: Y_{jj} > \alpha \phi^2 \sqrt{\frac{\log{n}}{m}}\right\},
\end{equation*}
where  $\alpha$ is a tuning parameter.  And then the initial guess is obtained by computing the singular vector of $\Y_{\widehat{\mcalS}_0,\widehat{\mcalS}_0}$. They showed that with $m \ge O(s^2\log{n})$ measurements, this approach yields a good initial estimate with high probability.

In this paper, we adopt the hard thresholding operator from \cite{jagatap2019sample} to estimate the support set
and introduce an alternative initialization method, as outlined in Algorithm \ref{algorithm1:Initialization}. Specifically, we first identify the indices of the largest $s$ diagonal elements of $\Y$ to obtain the estimated support set $\widehat{\mcalS}$. Next, we construct the initial estimate
$\x^0$ as the leading left singular vector of $\Y_{\widehat{\mcalS},\widehat{\mcalS}}$. Finally,  we scale $\x^0$ to have norm  $\phi$, ensuring it matches $\norm{\x}$.

\begin{algorithm}[H]
	\caption{Sparse Gauss-Newton Algorithm}
	\label{algorithm2:Iterative Algorithm}
	\renewcommand{\algorithmicrequire}{\textbf{Input:}}
	\renewcommand{\algorithmicensure}{\textbf{Output:}}
	\begin{algorithmic}[1]
		\REQUIRE Gaussian random matrix $\A_{i} \in \mbbR^{m \times n}, \ i=1,...,m$, measurement vector $\y \in \mbbR^m$, initial estimate $\x^0$, \ sparsity level $s$, and step size $\mu^k$.
		\FOR{$k=0,1,...$}
		\STATE $\u^k = \mathcal{H}_s(\x^k-\mu^k \nabla f(\x^k))$
		\STATE $\mcalS_{k+1} = {\rm supp}(\u^k)$
		\STATE  Update $\x^{k+1}$:
		\begin{eqnarray*}
			\x^{k+1}_{\mcalS_{k+1}}= \x^k_{\mcalS_{k+1}}-\p_{\mcalS_{k+1}}^k, \qquad \text{and} \qquad \x^{k+1}_{\mcalS_{k+1}^c}=0,
		\end{eqnarray*}
		where $\p_{\mcalS_{k+1}}^k$ is given in \eqref{eq:defpk}.
		\ENDFOR
		\ENSURE $\x=\x^{k+1}$
	\end{algorithmic}
\end{algorithm}

\subsection{Refinement}
Recall that the problem we aim to solve is
\begin{equation} \label{least1}
	\min\limits_{\z \in \mbbR^n, \norm{\z}_0 \le s} f(\z) := \frac{1}{2m}\sum\limits_{i=1}^{m}(\z^{\TT}\A_i\z-y_i)^2.
\end{equation}
In this paper, we incorporate the hard-thresholding operator and the Gauss-Newton method \cite{gao2017phaseless, blumensath2008iterative, foucart2011hard} to develop an efficient algorithm for recovering the target signal $\x$. Specifically, at the $k$-th iteration, an estimated support set $\mcalS_{k+1} $ is selected based on  the Gradient Hard Thresholding Pursuit method \cite{xyuan}, namely, 
\begin{equation*}
	\mcalS_{k+1} = {\rm supp}(\mathcal{H}_s(\x^k-\mu^k \nabla f(\x^k))),
\end{equation*}
where $\mathcal{H}_s$ is the hard thresholding operator that selects the $s$-largest values of a vector while setting the remaining elements to $0$.   Next, we restrict problem (\ref{least1}) to  the support set $\mcalS_{k+1}$:
\begin{equation} \label{eq:subop}
	\min_{{\rm supp}(\z) \subseteq \mcalS_{k+1}} \frac{1}{2m}\sum\limits_{i=1}^{m}(\z^{\TT}\A_i\z-y_i)^2.
\end{equation}
Finally, the  Gauss-Newton method is applied to solve \eqref{eq:subop} to obtain the next iteration $\x^{k+1}$.  Let
\[
F_{i}(\z) = \frac{1}{\sqrt{m}}(\z^{\TT}\A_i\z-y_i), i=1,\ldots,m.
\]

The first-order Taylor expansion of $F_i(\z)$ at $\x^k$ is
\begin{equation*}
	F_i(\z) \approx F_i(\x^k)+\nabla F_i(\x^k)^{\TT}(\z-\x^k).
\end{equation*}
Hence, \eqref{eq:subop}  can be approximated by
\begin{equation}
	\min_{{\rm supp}(\z) \subseteq \mcalS_{k+1}} \frac{1}{2}\norm{\F(\x^k)+\J(\x^k)(\z-\x^k)}^2,
	\label{approx least squares}
\end{equation}
where the $i$-th component of $\F(\x^k) \in \mbbR^{m}$ is $F_i(\x^k)$ and  the Jacobian matrix $\J(\x^k) \in \R^{m\times n}$ is defined as
\begin{equation} \label{eq:Jx}
	\J(\x^k)=\frac{1}{\sqrt{m}}[(\A_1+\A_1^{\TT})\x^k,...,(\A_m+\A_m^{\TT})\x^k]^{\TT}.
\end{equation}
For simplicity, we define $\textstyle \widetilde{\A}_i=\A_i+\A_i^{\TT}$.
Suppose that $\x^{k+1}$ is the solution to (\ref{approx least squares}), the first-order optimality condition yields that  $\x_{\mcalS_{k+1}^c}^{k+1} = 0$ and  $\x_{\mcalS_{k+1}}^{k+1}$ satisfies 
\begin{eqnarray*}
	&\J_{\mcalS_{k+1}}(\x^k)^{\TT}\J_{\mcalS_{k+1}}(\x^k)(\x_{\mcalS_{k+1}}^{k+1}-\x_{\mcalS_{k+1}}^k)
	\\ 
	&= \J_{\mcalS_{k+1}}(\x^k)^{\TT}\J_{\mcalS_{k+1}^c}(\x^k)\x^k_{\mcalS_{k+1}^c}-\J_{\mcalS_{k+1}}(\x^k)^{\TT}\F_{\mcalS_{k+1}}(\x^k).
\end{eqnarray*}

Observe that 
\begin{equation} \label{eq:gradf}
	\J(\x^k)^{\TT}\F(\x^k)=\nabla f(\x^k)=\frac{1}{m}\sum\limits_{i=1}^{m}((\x^k)^{\TT}\A_i\x^k-y_i)\widetilde{\A}_i\x^k.
\end{equation}
As a result, the next iteration $\x^{k+1}$ is 
\begin{equation} \label{eq:updarule}
	\x^{k+1}_{\mcalS_{k+1}}=\x^k_{\mcalS_{k+1}}- \p_{\mcalS_{k+1}}^k
\end{equation}
with 
\begin{equation} \label{eq:defpk}
	\p_{\mcalS_{k+1}}^k= \left[\J_{\mcalS_{k+1}}(\x^k)^{\TT}\J_{\mcalS_{k+1}}(\x^k)\right]^{-1}\Big(\nabla_{\mcalS_{k+1}} f(\x^k) -\J_{\mcalS_{k+1}}(\x^k)^{\TT}\J_{\mcalS_{k+1}^c}(\x^k)\x^k_{\mcalS^c_{k+1}}\Big),
\end{equation}
and
\begin{equation*}
	\x^{k+1}_{\mcalS_{k+1}^c}=0.
\end{equation*}
The specific details of the refinement are shown in Algorithm \ref{algorithm2:Iterative Algorithm}.

\section{Convergence guarantee} \label{sec:4}
%In this section, we provide the theoretical analysis of our algorithm for both noise-free and noisy cases.

In real-world applications, measurements are often corrupted by noise. Therefore, it is crucial to analyze the robustness of our algorithm in noisy settings. In noisy case, the measurements are 
\begin{equation}
	y_i = \x^{\TT}\A_i\x+ \epsilon_i, i=1,\ldots,m,
	\label{noisymeasure}
\end{equation}
where $\{\epsilon_i\}_{i=1}^{m}$ represent the noises. Throughout this paper, we assume that $\{\epsilon_i\}_{i=1}^{m}$ are independent centered  sub-exponential random variables with the maximum sub-exponential norm $\sigma$, i.e.,
\[
\sigma:=\mbox{max}_{1\le i\le m} \| \epsilon_i \|_{\psi_1}.
\] 
The following theorem demonstrates that  the initial guess  produced by Algorithm \ref{algorithm1:Initialization} is close enough to the true signal with high probability. 

\begin{theorem} \label{Theorem1}
	Let $\{\A_i\}_{i=1}^m$ be i.i.d. standard Gaussian random matrices with $\A_i \sim \mcalN^{n\times n}(0,1)$, and  $\x \in \mbbR^n$ be any fixed signal with $\norm{\x}_0 \le s$. For any $ \delta_0 \in (0,1)$, with probability at least $1-8n^{-5}-4\exp(-c\delta_0^2 m)-2\exp(-c's)$, the output of Algorithm \ref{algorithm1:Initialization} with $y_i = \x^{\TT}\A_i\x +\epsilon_i, i=1,\ldots, m$, obeys 
	\begin{equation*}
		\mbox{\rm dist}(\x^{0},\x) \le \delta_{0} \norm{\x},
	\end{equation*}
	as long as $m > C(\delta_0) (1 +\frac{\sigma}{\norm{\x}^2})^2 s^2\log{(mn)}$, where $C(\delta_0)$ is a constant depending on $\delta_0$, $c$ and $c'$ are universal positive constant.
\end{theorem}

\begin{proof}
	See  Appendix \ref{sec:appA}.
\end{proof} 

Once the initial estimate $\x^0$ falls within a small neighborhood of either $\x$ or $-\x$, the following theorem shows that Algorithm \ref{algorithm2:Iterative Algorithm} achieves at least a linear convergence rate. Furthermore, in the noise-free case, after at most $O(\log{(\norm{\x}/x_{\min})})$ iterations, the algorithm attains a quadratic convergence rate. 

\begin{theorem} \label{Theorem3}
	Let $\{\A_i\}_{i=1}^m$ be i.i.d. standard Gaussian random matrices with $\A_i \sim \mcalN^{n\times n}(0,1)$, and  $\x \in \mbbR^n$ be any fixed signal with $\norm{\x}_0 \le s$.  Consider the noisy measurements (\ref{noisymeasure}). Assume that the initial guess $\x^0$ satisfying ${\rm dist}\left(\x^0,\x\right)\le \delta \norm{\x}$ for some  $\delta  \in (0,0.01)$. There exists some positive constants $c', c'', \underline\mu, \overline\mu$ such that if the step size $\mu^k \in (\frac{\underline\mu}{\norm{\x}^2}, \frac{\overline\mu}{\norm{\x}^2})$, then with probability greater than  $1-c''\exp(- c' s)$, the sequence $\{\x^k\}_{k \ge 1}$ generated by Algorithm \ref{algorithm2:Iterative Algorithm} with $y_i = \x^{\TT}\A_i\x +\epsilon_i, i=1,\ldots, m$ obeys
	\begin{equation*}
		{\rm dist}(\x^{k+1},\x) \le \rho' \cdot {\rm dist}(\x^k,\x) + \frac{C \sigma}{\norm{\x}}  \sqrt{ \frac{s \log(en/s)}m}, \ \ \forall \; k\ge 0,
	\end{equation*}  
	provided  $m > C' s\log{n}$. Here,   $\rho' \in (0,1)$ and $C, C'$ are positive constants. Furthermore, when $\sigma = 0$, the sequence $\{\x^k\}_{k \ge 1}$ generated by Algorithm \ref{algorithm2:Iterative Algorithm} obeys 
	\begin{equation*}
		{\rm dist}(\x^{k+1},\x) \le 
		\begin{cases}
			\nu \cdot {\rm dist}(\x^k,\x),  & \forall \ k\ge 0,  \\
			\nu' \cdot {\rm dist}^2(\x^k,\x), & \forall \  k \ge  c_1\log{\left(\frac{\norm{\x}}{x_{\min}}\right)}+c_2. 
		\end{cases}
	\end{equation*}
	Here, $x_{\min}$ is the smallest nonzero entry in magnitude of $\x$, and $c_1$, $c_2$, $\nu \in (0,1)$, $\nu'$ are positive constants. 
\end{theorem}
\begin{proof}
	The proof of this theorem is deferred to Section \ref{sec:Thnoisy}.
\end{proof}
\begin{remark}
	The error bound $O\left( \frac{\sigma}{\norm{\x}}  \sqrt{ \frac{s \log(en/s)}m}\right)$ established in Theorem \ref{Theorem3}  is better than the results given in \cite{cai2023fast, Dai}, and  is minimax optimal \cite{Rigollet,Ye}.
\end{remark}

\begin{remark}
	Notably, when properly initialized, our refinement algorithm (Algorithm \ref{algorithm2:Iterative Algorithm}) achieves a quadratic convergence rate with $O(s\log{n})$ measurements with no need of sample splitting. This marks a significant improvement over existing second-order sparse phase retrieval methods, which either depend heavily on artificial resampling techniques \cite{Dai} or require a substantially higher measurement complexity of $O(s^2\log{n})$ \cite{cai2023fast,cai2016optimal}. 
\end{remark}

When $\sigma = 0$, combining Theorem \ref{Theorem1} and Theorem \ref{Theorem3} together, we can achieve an $\epsilon$-accurate solution within $O(\log(\log{(1/\epsilon)})+\log{(\norm{\x}/x_{\min})})$ iterations, as stated below.
\begin{corollary}
	Let $\{\A_i\}_{i=1}^m$ be i.i.d. standard Gaussian random matrices with $\A_i \sim \mcalN^{n\times n}(0,1)$, and  $\x \in \mbbR^n$ be any fixed signal with $\norm{\x}_0 \le s$. Assume that the initial guess $\x^0$ is generated by Algorithm \ref{algorithm1:Initialization}.  There exists some positive constants $C, C_1, C_2, C_3, \epsilon \in (0,1), \underline\mu, \overline\mu$ such that if the step size $\mu^k \in (\frac{\underline\mu}{\norm{\x}^2}, \frac{\overline\mu}{\norm{\x}^2})$, then with probability at least  $1-C_1\exp(- C_2 m) - C_3n^{-5}$, the output $\x^K$ of Algorithm \ref{algorithm2:Iterative Algorithm} with $y_i = \x^{\TT}\A_i\x, i=1,\ldots, m$ obeys 
	\begin{equation*}
		{\rm dist}(\x^K,\x)<\epsilon\norm{\x},
	\end{equation*}  
	provided $m\ge Cs^2 \log(mn)$. Here,  $K$ is an integer obeying $K \le O(\log(\log{(1/\epsilon)})+\log{(\norm{\x}/x_{\min})})$.
\end{corollary}

\section{Numerical experiments} \label{sec:5}

In this section, we present the results of several numerical experiments to evaluate the empirical performance of the algorithm proposed in this paper. All experiments were conducted on a desktop computer running MATLAB 2020a, equipped with an Intel CPU at 3.00 GHz and 32 GB of RAM.

\begin{figure}[!t]
	\centering
	\includegraphics[width=3.5in]{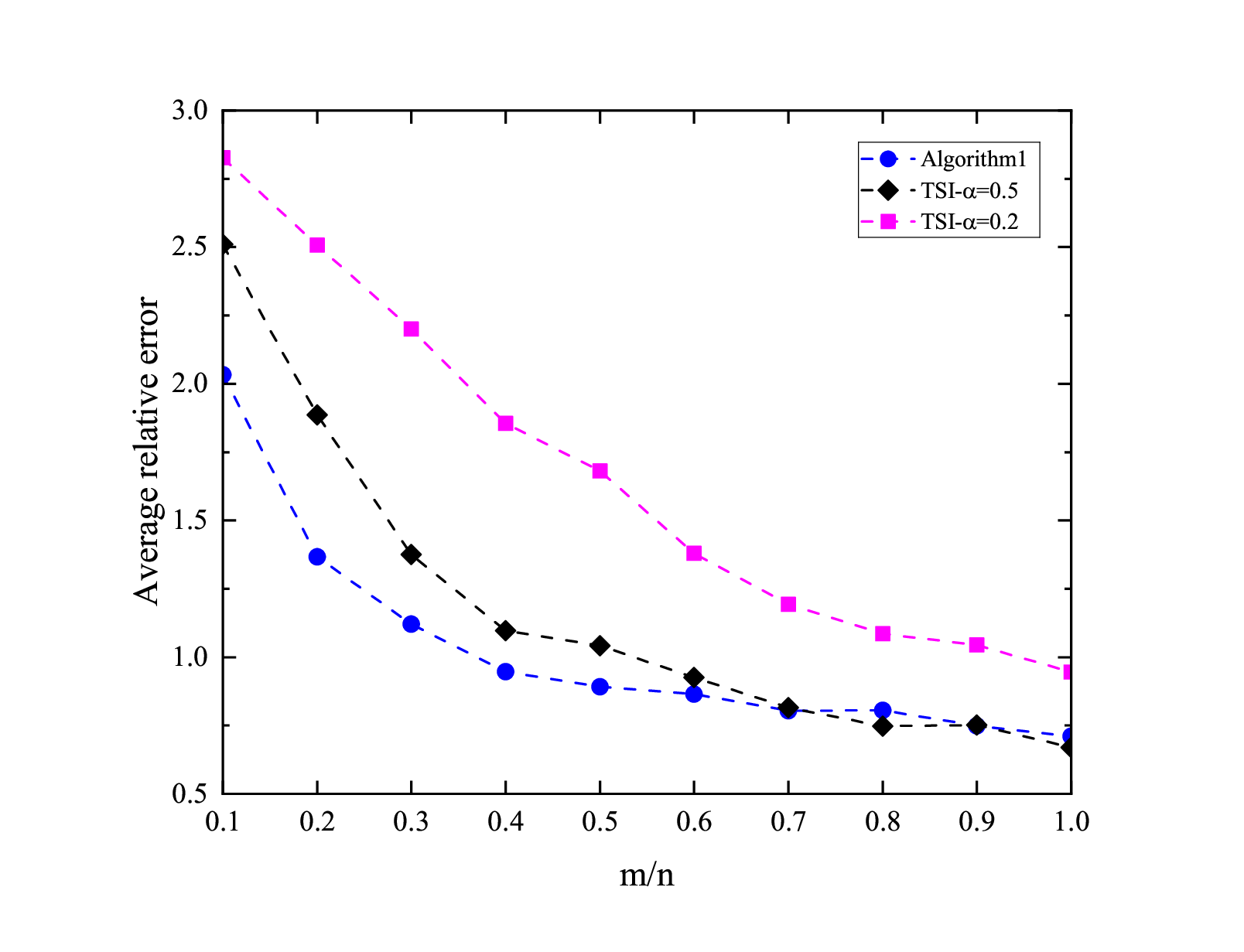}
	\caption{Initialization experiments: Average relative error between $\x^0$ and $\x$ for $n=500$, $s=5$ and $m/n$ varying within the range $[0.1,1]$.}
	\label{fig_1}
\end{figure}

Throughout the experiments, we generate each entry of the measurement matrices $\A_i \in \mbbR^{n \times n}, i = 1,... ,m$, independently from the standard normal distribution $\mathcal{N}(0,1)$. The target signal $\x \in \mbbR^n$ is constructed as an 
$s$-sparse vector whose non-zero entries are drawn i.i.d. from $\mathcal{N}(0,1)$. We consider both noise-free measurements 
$y_i = \x^{\TT}\A_i\x$ and noisy measurements $y_i = \x^{\TT}\A_i\x+\epsilon_i$, where $\epsilon_i$ represents additive noise and follows centered sub-exponential distribution.
To quantify reconstruction accuracy, we define the relative error  between the estimated signal $\hat{\x}$ and the ground truth $\x$ as:
\begin{equation*}
	{\rm Relative \ error} := \frac{{\rm dist}(\hat{\x},\x)}{\norm{\x}} = \frac{\min\{\norm{\hat{\x}-\x},\norm{\hat{\x}+\x}\}}{\norm{\x}}.
\end{equation*}

{\em Experiment 1: Comparative evaluation of initialization methods.}

In this experiment, we conduct a systematic comparison between our proposed initialization method and the Thresholded Spectral Initialization (TSI) approach from \cite{chen2025solving}.  We fix the signal dimension $n=500$ and sparsity level $s=5$. The measurement ratio $m/n$ is varied from 0.1 to 1 with step size being 0.1.  For each $m$, we conduct 100 independent trials and compute average relative errors.   While implementing TSI with its recommended parameter $\alpha = 0.5$ \cite{chen2025solving}, we additionally examine $\alpha = 0.2$ to assess parameter sensitivity. As demonstrated in Fig. \ref{fig_1}, our method consistently achieves superior accuracy, with reconstruction errors significantly smaller than those of TSI variants across all measurement regimes.
Notably, our proposed initialization algorithm has an average relative error that is approximately $20\%$ lower than that of the TSI algorithm with the optimal parameter settings as $m/n \le 0.3$. When $m/n$ exceeds $0.3$, the proposed algorithm's performance begins to degrade but remains favorable. In particular, at a sampling ratio of $m/n=0.5$, our algorithm achieves an average relative error of 0.8, while TSI-$\alpha=0.5$ and TSI-$\alpha = 0.2$ are $1.0$ and $1.6$, respectively.

{\em Experiment 2: Phase transition.}
\begin{figure*}[t]
	\centering
	\subfloat[]{\includegraphics[width=0.33\textwidth]{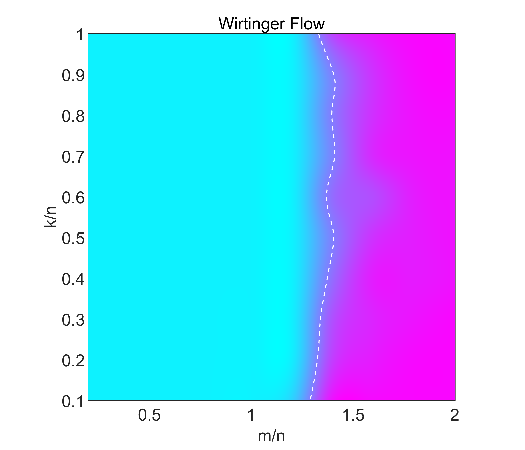}%
		\label{fig_first_case}}
	\hfil
	\subfloat[]{\includegraphics[width=0.33\textwidth]{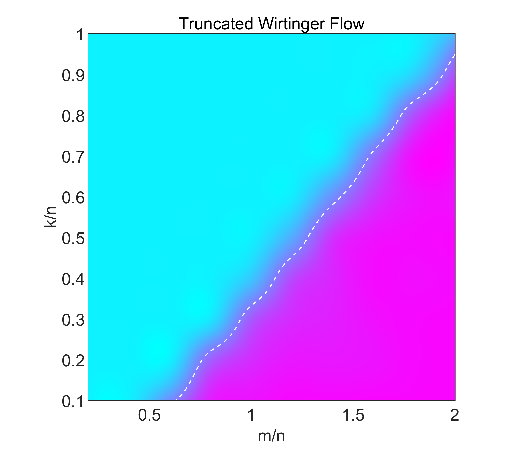}%
		\label{fig_second_case}}
	\subfloat[]{\includegraphics[width=0.33\textwidth]{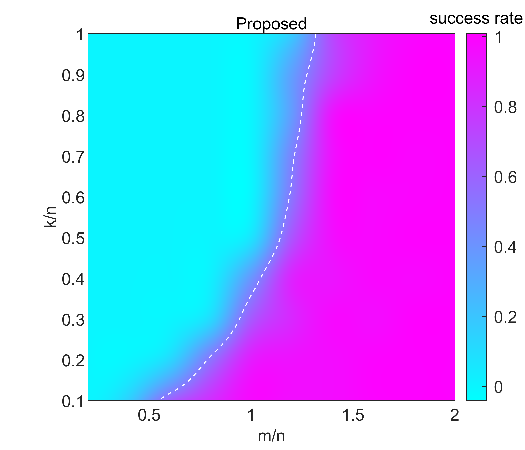}%
		\label{fig_third_case}}
	\caption{The phase transitions of WF, TWF, and proposed algorithm. }
	\label{fig_2}
\end{figure*}

In this experiment, we illustrate the phase transition of Algorithm 2 and compare its performance with WF \cite{huang2019solving} and  TWF \cite{chen2025solving}.  The signal dimension is fixed to $n=100$, while the sparsity level varies from $0.1n$ to $n$ with step size being $0.1n$. The number of measurements $m$ ranges from $0.2n$ to $2n$ with step size $0.2n$.  All settings for WF and TWF follow those specified in the respective authors' papers.
The maximum number of iterations for all algorithms is set to 2000.
We consider a trial to be successful when the relative error $\mbox{dist}(\hat{\x},\x)/\norm{\x} < 10^{-3}$.
For each combination of $m$ and $s$, we perform 100 independent trials to calculate the success rate. 
All algorithms are initialized with the same initial guess $\x^0$ generated by Algorithm 1. 
The phase transition of WF, TWF, and Algorithm 2 are shown in Fig. \ref{fig_2}. Comparing Fig. \ref{fig_2}(a) with Fig. \ref{fig_2}(b) and Fig. \ref{fig_2}(c), we observe that algorithms leveraging sparse priors significantly reduce sample complexity.
More specifically, when $m/n<1.4$, we can see that the success rate achieved by the WF algorithm under all levels of sparsity is less than $50\%$. However, the TWF algorithm and the algorithm proposed in this paper that utilizes the sparse prior when the sparsity is relatively low $(k/n \le 0.4)$, only require a sampling ratio of $m/n>0.5$ to achieve a success rate of over $50\%$.
Furthermore, when the sparsity level is relative high $(k/n>0.4)$, the proposed algorithm requires fewer sampling ratio to achieve the same success rate.
It means that our proposed algorithm outperforms WF and TWF in terms of sampling complexity, demonstrating a substantial improvement.

{\em Experiment 3 : Convergence behavior.}

We compare the convergence behavior of Algorithm 2 with TWF algorithm.
We plot the mean relative error versus iteration count in Fig. \ref{fig3}. Here we set the signal dimension $n$, measurements $m$ and sparsity $s$ to be $(200,200,40)$. The maximum number of iterations of the algorithms is all set to 1000. We conduct 100 independent trials and calculate the average relative error. Fig. \ref{fig3} shows that after the same number of iterations, our proposed iterative algorithm can achieve more accurate recovery.
\begin{figure}[!t]
	\centering
	\includegraphics[width=3.5in]{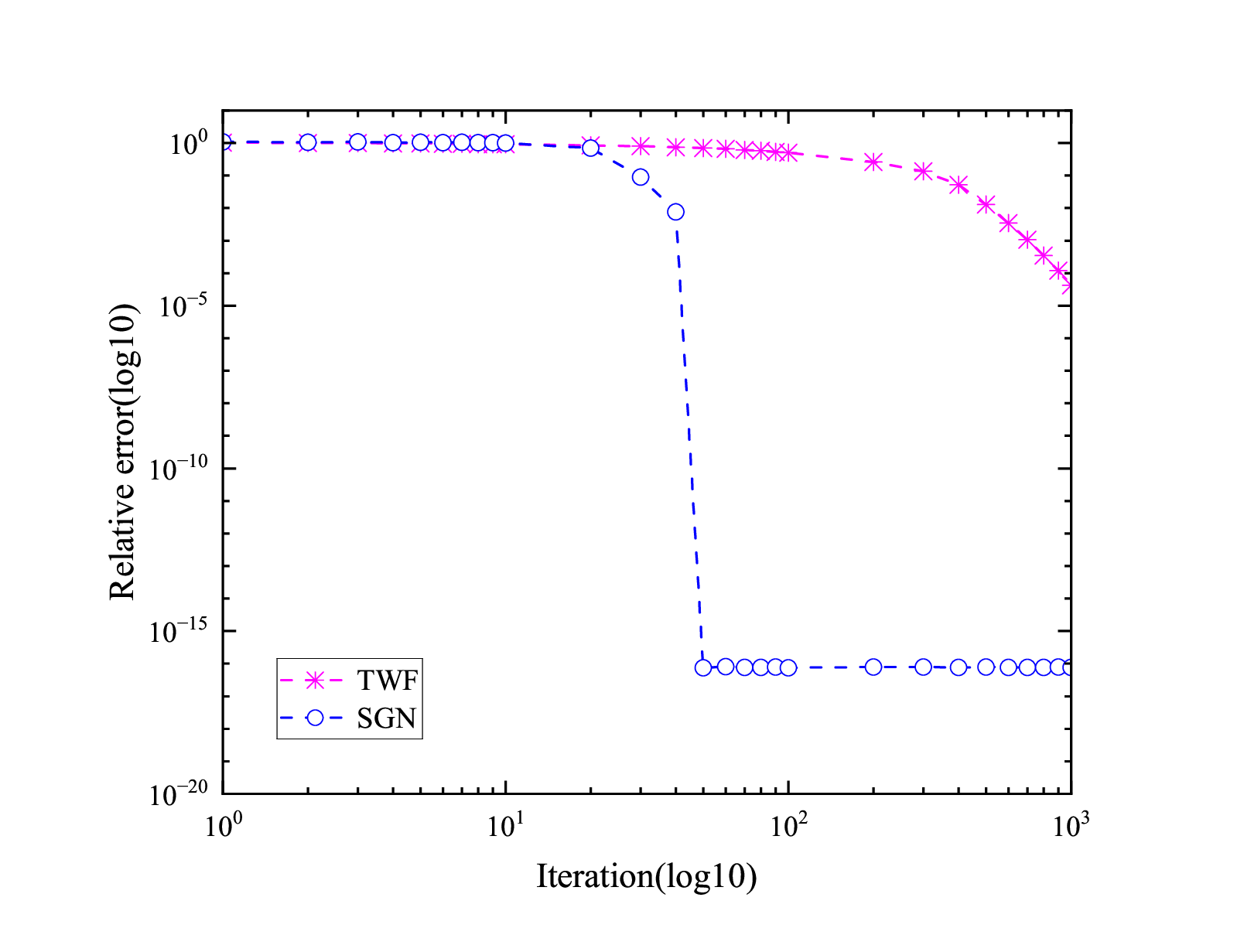}
	\caption{Convergence behavior: Plot of relative errors versus number of iterations for the proposed algorithm (SGN) and TWF.}
	\label{fig3}
\end{figure}
\begin{figure}[!t]
	\centering
	\begin{minipage}[b]{\linewidth}
		\subfloat[$m=\lfloor 10s\log{n}\rfloor$]{
			\includegraphics[width=0.48\linewidth]{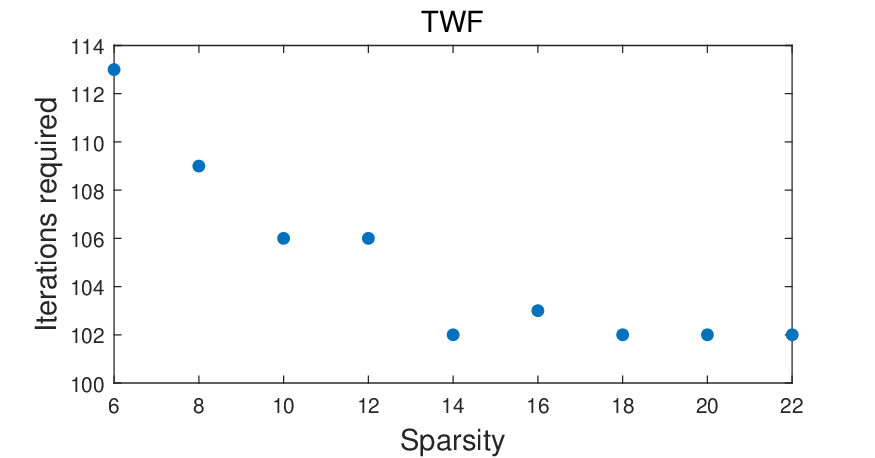}
		}
		\subfloat[$m=\lfloor 10s\log{n}\rfloor$]{
			\includegraphics[width=0.48\linewidth]{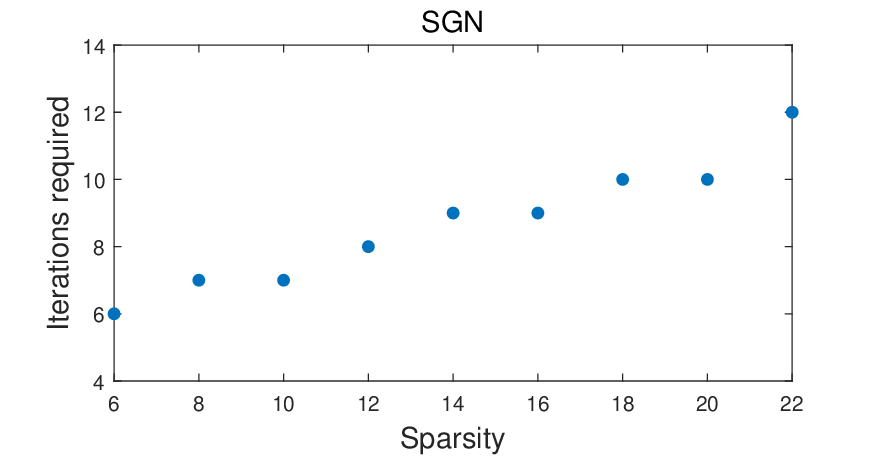}
		}
	\end{minipage}
	
	\caption{The average number of iterations required when successful recovery is achieved.}
	\label{fig4}
\end{figure}

In Fig. \ref{fig4}, we plot the average number of iterations required for the algorithm to reach successful recovery (i.e. ${\rm dist}(\x^k,\x)/\norm{\x} \le 10^{-3}$) at different sparsity levels. Specifically, we fix the signal dimension $n=100$, and set sample size $m=\lfloor 10s\log{n}\rfloor$, where $s$ is the sparsity. The sparsity $s$ varies from 6 to 22 with step size 2. 
For each sparsity, we conduct 100 independent trials. We  record the average number of iterations required for successful recovery. As can be seen from Fig. \ref{fig3} and Fig. \ref{fig4}, our proposed algorithm only requires $1/10$ iterations of the existing algorithm to achieve convergence and reach the lower relative error.

\begin{figure}[H]
	\centering
	\includegraphics[width=3.5in]{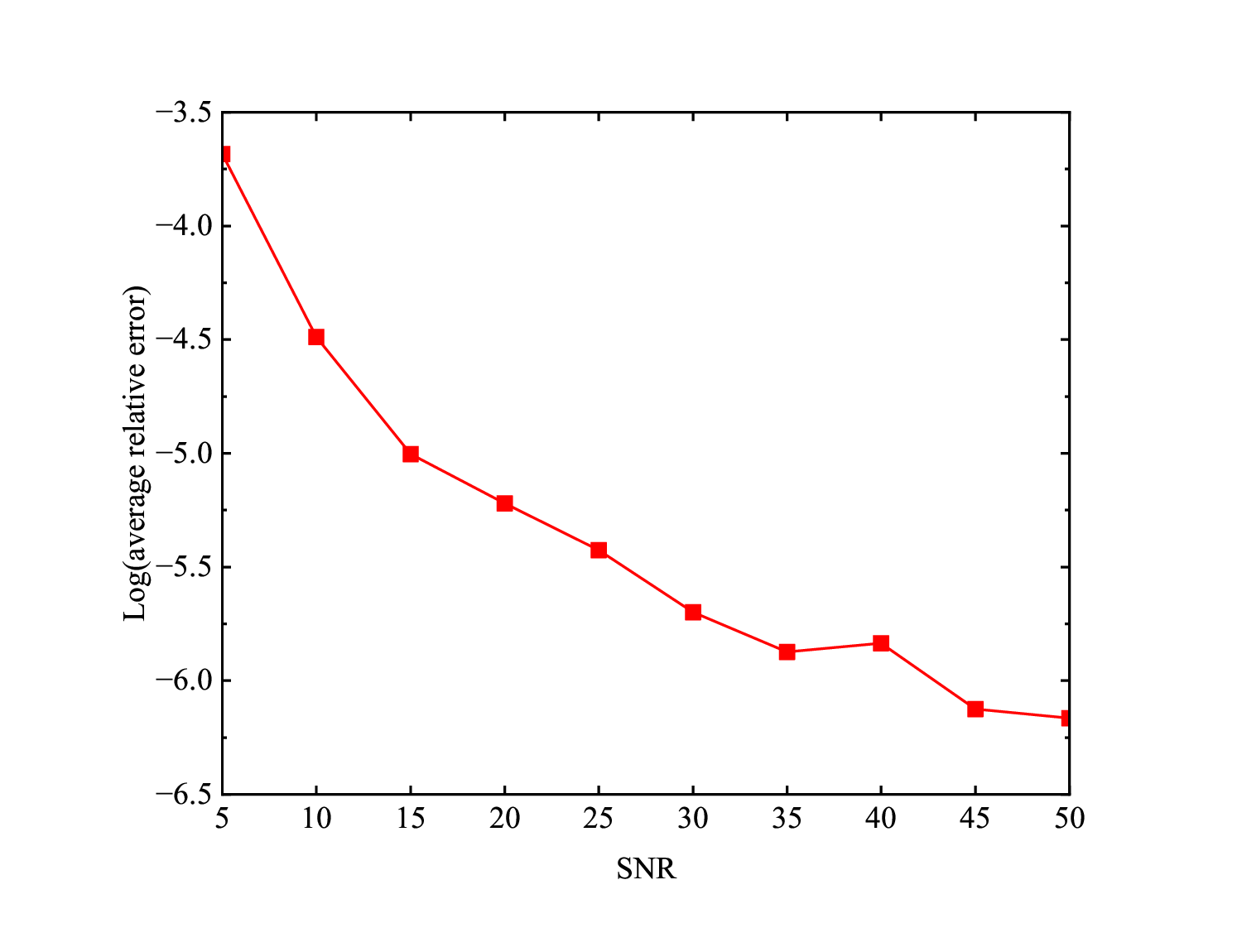}
	\caption{Robustness of SGN to additional Gaussian noise}
	\label{fig_10}
\end{figure}

{\em Experiment 4: Robustness to noise.}

In this experiment, we investigate the robustness of the proposed SGN algorithm to additive Gaussian noise. 
The experiment result is shown in Fig. \ref{fig_10}. The measurements are generated by $y_i = \x^{\TT}\A_i\x+\epsilon_i$, where $\epsilon_i$ follows independent Gaussian distribution $\mathcal{N}(0,\sigma^2)$. 
We set the signal dimension $n=100$, the number of measurements $m=200$, and the sparsity $s=5$. The noise strength is quantified by the signal-to-noise ratio (SNR), defined as $\norm{\x}/\sigma^2$. 
The SNR is varied from 5 to 50 in steps of 5.
For each SNR, we perform 50 independent experiments, and compute the average relative error. It is evident from Fig. \ref{fig_10} that the average relative error shows a continuous decreasing trend as the SNR increases.
When the SNR is relative small, the algorithm can still achieve a log relative error of at most -3.5, which indicates that the proposed algorithm exhibits excellent robustness to Gaussian noise.

\section{Conclusion} \label{sec:6}
This paper develops a novel two-stage algorithm for solving sparse quadratic systems  $y_i=\x^{\TT}\A_i\x$ under the Gaussian measurements. Theoretically, we prove that the proposed algorithm achieves a quadratic convergence rate after a finite number of iterations, with a sampling complexity of $m=O(s^2\log{n})$. In the presence of noise, our algorithm remains robust and attains the optimal error bound. Numerical experiments indicate that the proposed algorithm outperforms state-of-the-art methods in terms of convergence speed and sampling efficiency. While the initialization stage currently limits sample efficiency, our findings highlight the potential for further improvements through more advanced initialization schemes, offering a promising direction for future research in high-dimensional quadratic inverse problems.

\section{Appendix: Proofs of main results} \label{sec:diszkl} \label{sec:7}
\subsection{ Proof of Theorem \ref{Theorem4}}
\label{sec: realmeas}
\begin{proof}
	Note that the set of all $\{\A_i\}_{i=1}^m \in \mathbb R^{m\times n^2}$ has real dimension $mn^2$. From the definition of generic, to prove the result, we only need to show that the set of $\{\A_i\}_{i=1}^m$ such that $\mathbf{M}_{\mathcal A}$ is not injective on $\Sigma_s ({\mathbb R}^n)$ lies in a finite union of subsets with dimension strictly less than $mn^2$.

	Note that $\mathbf{M}_{\mathcal A}(\x)=\mathbf{M}_{\mathcal A}(\z)$ if and only if $\mathbf{M}_{\mathcal A}(a\x)=\mathbf{M}_{\mathcal A}(a\z)$  for any $a \in \R$. Therefore, for any subset of indices $I, J\subset [n]$ with $|I|\le s$, $|J|\le s$, we define  $\mathcal G_{I,J}$ be the set of all $(\A_1,\ldots,\A_m, \x,\z) \in \mathbb P(\R^{n \times n} \times \cdots \times \R^{n \times n}) \times \mathbb P(\R^n \times \R^n)$ such that $\x\neq \pm \z$ and
	\[
	\mbox{supp}(\x) \subset I, \quad \mbox{supp}(\z) \subset J,  \quad \mathbf{M}_{\mathcal A}(\x)=\mathbf{M}_{\mathcal A}(\z).
	\]
	Then $\mathcal G_{I,J}$ is a well-defined real projective variety because the defining equations are homogeneous polynomials.
	We next consider the dimension of this projective variety $\mathcal G_{I,J}$. To this end, let $\pi_1$ be projections on the first $m$ coordinates of $\mathcal G_{I,J}$ and $\pi_2$ be projections on the last two coordinates of $\mathcal G_{I,J}$, namely, 
	\[
	\pi_1(\A_1,\ldots,\A_m, \x,\z)=(\A_1,\ldots,\A_m) \qquad \mbox{and} \qquad \pi_2(\A_1,\ldots,\A_m, \x,\z)=(\x,\z).
	\]
	For any $\x_0 \neq \pm \z_0$ and any fixed $j$, the equation $\x_0^{\TT}\A_j\x_0 =\z_0^{\TT}\A_j\z_0$ defines a nonzero homogeneous polynomial, which constraints the set of $\A_j$ to lie on a real projective variety of codimension $1$. Note that there are $m$  such constrains.
	Hence, for any $\x_0 \neq \z_0$ in $\R^n$ with $\mbox{supp}(\x_0) \subset I$, $ \mbox{supp}(\z_0) \subset J$, the preimage $\pi_2^{-1}(\x_0,\z_0)$ has dimension $mn^2-m-1$. It follows from \cite[Cor. 11.13]{Fiber} that 
	\[
	\mbox{dim}(\mathcal G_{I,J})\le \mbox{dim}(\pi_2^{-1}(\x_0,\z_0))+|I|+|J|-1\le mn^2-m-1+2s-1.
	\]
	If $m\ge 2s$, then
	\[
	\mbox{dim}(\pi_1(\mathcal G_{I,J})) \le \mbox{dim}(\mathcal G_{I,J}) \le mn^2-2<mn^2-1=\mbox{dim}(\mathbb P(\R^{n \times n} \times \cdots \times \R^{n \times n})).
	\]
	Note that $\pi_1(\mathcal G_{I,J})$ contains precisely those $(\A_1,\ldots,\A_m) \in \R^{n \times n} \times \cdots \times \R^{n \times n}$ such that $\mathbf{M}_{\mathcal A}$ is not injective for signals $\x,\z$ with $\mbox{supp}(\x) \subset I$, $\mbox{supp}(\z) \subset J$. Thus, a generic $(\A_1,\ldots,\A_m) \in \R^{n \times n} \times \cdots \times \R^{n \times n}$ guarantees the injective for all $\x,\z$ with $\mbox{supp}(\x) \subset I$, $\mbox{supp}(\z) \subset J$. 
	
	Finally, there are only finitely many indices subsets $I,J$. The theorem is then proved.
\end{proof}

\subsection{ Proof of Theorem \ref{Theorem5}}
\label{sec: complexmeas}
\begin{proof}
	The proof is similar to that in the real case. The set of all $\{\A_i\}_{i=1}^m \in \mathcal H^{n\times n} \times \cdots \times \mathcal H^{n\times n}$ has real dimension $\mbox{dim}_{\R}(\mathcal H^{n\times n} \times \cdots \times \mathcal H^{n\times n})=mn^2$. Next, we  show that the set of $\{\A_i\}_{i=1}^m$ such that $\mathbf{M}_{\mathcal A}$ is not injective on $\Sigma_s ({\mathbb C}^n)$ lies in a finite union of subsets with dimension strictly less than $mn^2$.

	Note that $\mathbf{M}_{\mathcal A}(\x)=\mathbf{M}_{\mathcal A}(\z)$ if and only if $\mathbf{M}_{\mathcal A}(a\x)=\mathbf{M}_{\mathcal A}(a\omega \z)$  for any $a \in \R$ and $\omega\in \mathbb C$ with $|\omega|=1$. Therefore, for any $\x,\z \in \mathbb C^n$, if $\mathbf{M}_{\mathcal A}(\x)=\mathbf{M}_{\mathcal A}(\z)$, then we can find $\x_0, \z_0 \in \mathbb C^n$ such that the following holds:
	\begin{itemize}
		\item[(i)] $\mathbf{M}_{\mathcal A}(\x_0)=\mathbf{M}_{\mathcal A}(\z_0)$;
		\item[(ii)] The first nonzero entry of $\x_0$ is $1$;
		\item[(iii)] The first nonzero entry of $\z_0$ is real. 
	\end{itemize}
	For convenience, we denote $X$ the  subset of $\mathbb C^n$ consisting of elements $\x \in \mathbb C^n$ whose first nonzero entry is $1$, and  denote $Z$ the subset of $\mathbb C^n$ consisting of elements $\z\in \mathbb C^n$ whose first nonzero entry is real.  Moreover, let $\mathbb C_{I}^n$ denote the set of vectors $\x \in \mathbb C^n$ such that $\mbox{supp}(\x) \subset I$.
	Now,  for any subset of indices $I, J\subset [n]$ with $|I|\le s$, $ |J|\le s$, we define  $\mathcal G_{I,J}$ be the set of all $(\A_1,\ldots,\A_m, \x,\z) \in \mathcal H^{n\times n} \times \cdots \times \mathcal H^{n\times n} \times \mathbb{C}^n \times \mathbb{C}^n$ such that $\x \neq  \omega \z$ for any $\omega\in \mathbb C$ with $|\omega|=1$ and
	\[
	\x \in X \cap \mathbb C_{I}^n, \quad \z \in Z \cap \mathbb C_{J}^n, \quad  \mathbf{M}_{\mathcal A}(\x)=\mathbf{M}_{\mathcal A}(\z).
	\]
	Let $\pi_1$ be projections on the first $m$ coordinates of $\mathcal G_{I,J}$ and $\pi_2$ be projections on the last two coordinates of $\mathcal G_{I,J}$, namely, 
	\[
	\pi_1(\A_1,\ldots,\A_m, \x,\z)=(\A_1,\ldots,\A_m) \qquad \mbox{and} \qquad \pi_2(\A_1,\ldots,\A_m, \x,\z)=(\x,\z).
	\]
	For any fixed $\x_0, \z_0 \in \mathbb C^n$ and each fixed $j$, if $\x_0 \neq  \omega \z_0$ for any $\omega\in \mathbb C$ with $|\omega|=1$, then the equation $\x_0^{\TT}\A_j\x_0 =\z_0^{\TT}\A_j\z_0$ defines a nonzero homogeneous polynomial, which constraints the set of $\A_j$ to lie on a real algebraic variety of codimension $1$. Hence, for any $\x_0 \neq  \omega \z_0$ in $\mathbb C^n$ with $ \x \in X \cap \mathbb C_{I}^n, \; \z \in Z \cap \mathbb C_{J}^n$, the preimage $\pi_2^{-1}(\x_0,\z_0)$ has real dimension 
	\[
	\mbox{dim}_{\R}(\pi_2^{-1}(\x_0,\z_0))=mn^2-m.
	\]
	It follows that 
	\begin{eqnarray*}
		\mbox{dim}_{\R}(\mathcal G_{I,J}) & \le &  \mbox{dim}_{\R}(\pi_2^{-1}(\x_0,\z_0))+\mbox{dim}_{\R}(X\cap \mathbb C_{I}^n)+\mbox{dim}_{\R}(Y \cup \mathbb C_{I}^n) \\
		&= &  mn^2-m+2s-2+2s-1.
	\end{eqnarray*}
	Therefore, if $m\ge 4s-2$, then
	\[
	\mbox{dim}_{\R}(\pi_1(\mathcal G_{I,J})) \le \mbox{dim}_{\R}(\mathcal G_{I,J}) \le mn^2-1<mn^2=\mbox{dim}_{\R}(\mathcal H^{n\times n} \times \cdots \times \mathcal H^{n\times n}).
	\]
	Since $\pi_1(\mathcal G_{I,J})$ contains precisely those $(\A_1,\ldots,\A_m) \in \mathcal H^{n\times n} \times \cdots \times \mathcal H^{n\times n}$ such that $\mathbf{M}_{\mathcal A}$ is not injective for signals $\x,\z$ with $\mbox{supp}(\x) \subset I$, $\mbox{supp}(\z) \subset J$.  Thus, a generic $(\A_1,\ldots,\A_m) \in \mathcal H^{n\times n} \times \cdots \times \mathcal H^{n\times n}$ guarantees the injective of $\mathbf{M}_{\mathcal A}$ for all $\x,\z$ with $\mbox{supp}(\x) \subset I$, $\mbox{supp}(\z) \subset J$. 
	
	Finally, there are only finitely many indices subsets $I,J$ with $|I|\le s, |J|\le s$. The theorem is  proved.
\end{proof}

\subsection{ Proof of Theorem \ref{Theorem1}} \label{sec:appA}
Before proving Theorem \ref{Theorem1}, we first introduce some auxiliary lemmas and notations.

\begin{definition}
	\cite[$\epsilon$-net]{vershynin2018high}
	Let $(X,d)$ be a metric space. $S$ is a subset of $X$ and let $\epsilon >0$. If $\mcalN \subseteq S$ satisfies
	\begin{equation*}
		\forall x \in S, \quad \exists \  x_0 \in \mcalN, \quad s.t. \quad d(x,x_0) \le \epsilon,
	\end{equation*}
	then we call $\mcalN$ an $\epsilon$-net of $S$. Here, $d(x,x_0)$ denotes the distance between $x$ and $x_0$.
\end{definition}

\begin{lemma}
	\cite[Product of sub-gaussians is sub-exponential]{vershynin2018high} Let $X$ and $Y$ be sub-gaussian random variables. Then $XY$ is sub-exponential. Moreover, 
	\begin{equation*}
		\norm{XY}_{\psi_1} \le \norm{X}_{\psi_2}\norm{Y}_{\psi_2}.
	\end{equation*}
\end{lemma}

\begin{lemma} \cite[Hoeffding’s inequality]{vershynin2018high} \label{le:Hoeff}
	Let $X_1,...,X_N$ be independent, mean zero, sub-gaussian random variables, and $a=(a_1,\ldots,a_N) \in \mbbR^N$ Then, for every $t \ge 0$, we have
	\begin{equation*}
		\mbbP\left(\abs{\sum\limits_{i=1}^{N}a_iX_i} \ge t\right) \le 2\exp\left(-\frac{Ct^2}{K^2\norm{a}^2}\right)
	\end{equation*}
	where $K=\max\limits_i \norm{X_i}_{\psi_2}$ and $C>0$ is a universal constant.
\end{lemma}

\begin{lemma} \cite[Bernstein’s inequality]{vershynin2018high} \label{le:bern}
	Let $X_1,...,X_N$ be independent, mean zero, sub-exponential random variables. Then, for every $t \ge 0$, we have
	\begin{equation*}
		\mbbP\left(\abs{\frac{1}{N}\sum\limits_{i=1}^{N}X_i} \ge t\right) \le 2\exp\left(-C \min\left(\frac{t^2}{K^2},\frac{t}{K}\right)N\right)
	\end{equation*}
	where $K=\max\limits_i \norm{X_i}_{\psi_1}$ and $C>0$ is a universal constant.
\end{lemma}

\begin{lemma}\cite[Weyl's inequality, singular values]{Weyl}
	\label{le: Weyl}
	Let $\M, \N \in \mathbb R^{m \times n}$ be two matrices with $m \le n$, The singular values of $\M$ and $\N$ are denoted as $\sigma_1(\M)\ge \ldots \ge \sigma_m(\M)$ and $\sigma_1(\N)\ge \ldots \ge \sigma_m(\N)$, respectively. Then we have
	\begin{eqnarray*}
		|\sigma_i(\M+\N)-\sigma_i(\M)| \le \sigma_1(\N) \quad \mbox{for} \quad 1 \le i \le m.
	\end{eqnarray*}
	Consequently, 
	\begin{eqnarray*}
		\sigma_i(\M)-\sigma_{i+1}(\M+\N) \ge \sigma_i(\M)-\sigma_{i+1}(\M)- \sigma_1(\N) \quad \mbox{for} \quad 1 \le i \le m.
	\end{eqnarray*}

	%	\begin{eqnarray*}
		%\sigma_{i+j-1}(\M+\N) \le \sigma_i(\M) + \sigma_j(\N)	%\end{eqnarray*}
		%	for $1 \le i,j \le m$ and $i+j \le m+1$. In particular,

	\end{lemma}

	\begin{lemma}	\cite[Lemma 4]{chen2025solving} \label{le:normphi}
		Let $\{\A_i\}_{i=1}^m$ be i.i.d. standard Gaussian random matrices with $\A_i \sim \mcalN^{n\times n}(0,1)$. 
		For any $\delta>0$, with probability at least $1-2\exp(-c\delta^2 m)$, it holds
		\[
		\abs{\frac1m \sum_{i=1}^m (\x^\TT \A_i \x)^2-\norm{\x}^4} \le \delta\norm{\x}^4.
		\]	
	\end{lemma}
	
	Recall that $\phi = (\frac{1}{m}\sum\limits_{i=1}^{m}y_i^2)^{1/4} = (\frac1m \sum\limits_{i=1}^m (\x^\TT \A_i \x)^2)^{1/4}$ and $\norm{\x^0}=\phi$, thus Lemma \ref{le:normphi} states that $(1-\delta)^{\frac{1}{4}}\norm{\x} \le \norm{\x^0} \le (1+\delta)^{\frac{1}{4}}\norm{\x}$ with probability ate least $1-2\exp(-c\delta^2m)$.

	\begin{lemma} \label{le:epAi}
		Suppose that $\{\epsilon_i\}_{i=1}^m$ are independent sub-exponential random variables with $\max\limits_i\norm{\epsilon_i}_{\psi_1}=\sigma$.  Assume that $\A_i \sim \mcalN^{n\times n}(0,1)$. Then with probability at least $1-2\exp(-c' s)$, it holds
		\[
		\norm{\frac{1}{m}\sum\limits_{i=1}^{m} \epsilon_i  [\A_i]_{\mcalS,\mcalS} } \le C_0 \sigma \sqrt{ \frac{s \log(en/s)}m}
		\]
		for all $\mcalS\subset [n]$ with $|\mcalS| \le 3s$. Here, $c'$ and $C_0$ are universal positive constants.
	\end{lemma}
	\begin{proof}
		For any fixed $\mcalS\subset [n]$ with $|\mcalS| \le 3s$,  define 
		\begin{equation*}
			\mcalX_{\mcalS} = \{\u \in \mbbS^{n-1}:{\rm supp}(\u) \subset \mcalS \},
		\end{equation*}
		where $\mathbb S^{n-1}:=\{\z \in \mathbb R^n : \norm{\z}=1\}$.
		Let $\mathcal{N}_{1/4}$ be a $1/4$-net of $\mcalX_{\mcalS} $ with cardinality $\abs{\mathcal{N}_{1/4}} \le 9^{3s}$. Then 
		\[
		\norm{\frac{1}{m}\sum\limits_{i=1}^{m} \epsilon_i  [\A_i]_{\mcalS,\mcalS} } \le 2\sup\limits_{\u,\v \in \mathcal{N}_{1/4}} \frac{1}{m}\sum\limits_{i=1}^{m} \epsilon_i  \u^\TT \A_i \v.
		\]
		Since $\u^\TT \A_i \v, i=1,\ldots, m$, are independent standard Gaussian random variables, Lemma \ref{le:Hoeff} gives
		\[
		\abs{\sum\limits_{i=1}^{m} \epsilon_i  \u^\TT \A_i \v} \le C_1 \norm{\bm \epsilon}\sqrt{s \log(en/s)} \le C_2 \sigma \sqrt{m s \log(en/s)}
		\]
		holds with probability at least $1-2\exp(-c_0 s \log (en/s))$. Here, we use \cite[Lemma A.7]{cai2016optimal}  in the last inequality, and $C_1, C_2>0, c_0>100$ are universal constants.  Combining the above estimators together and taking the union bound, we obtain that with probability at least $1-2\exp(-c'_0 s \log (mn))$, it holds
		\[
		\norm{\frac{1}{m}\sum\limits_{i=1}^{m} \epsilon_i  [\A_i]_{\mcalS,\mcalS} } \le 2 C_2 \sigma \sqrt{ \frac{s \log(en/s)}m}
		\]
		for any fixed $\mcalS\subset [n]$ with $|\mcalS| \le 3s$. Finally, taking the union bound over all  $\mcalS\subset [n]$ with $|\mcalS| \le 3s$ and noting that $\sum\limits_{k=1}^{3s} \left(\begin{array}{c} n\\k\end{array} \right) \le \left(en/s\right)^{3s}$, we arrive at the conclusion.
		
	\end{proof}

	\begin{lemma} \label{lemma2}
		Assume that $m\ge c \log{n}$ for some  constant $c>0$. 
		$
		y_i = \x^{\TT}\A_i\x + \epsilon_i, i = 1, \dots,m,
		$
		where the noise $\{\epsilon_i\}_{i=1}^m$ are independent sub-exponential random variables with \\ $\max\limits_i \norm{\epsilon_i}_{\psi_1}=\sigma$. Then with probability at least $1-4n^{-5}$,  it holds
		\begin{equation}
			\abs{Y_{jj}-x_j^2} \le C\left(1 +\frac{\sigma}{\norm{\x}^2}\right) \norm{\x}^2 \sqrt{\frac{\log{(mn)}}{m}} \quad \mbox{for} \quad j=1,\ldots,n.
			\label{eq:bound1}
		\end{equation}
		Here, $C>0$ is an absolute constant and $Y_{jj}$ is defined in Algorithm \ref{algorithm1:Initialization}.
		In particular, for all $j \in \mcalS^c$, here $\mcalS=\mbox{supp}(\x)$, with probability exceeding $1-4n^{-5}$, one has
		\begin{equation}
			Y_{jj}  \le C  \left(1 +\frac{\sigma}{\norm{\x}^2}\right) \norm{\x}^2 \sqrt{\frac{\log{(mn)}}{m}}. 
			\label{eq:bound2}
		\end{equation}
	\end{lemma}
	\begin{proof}  
		Denote the noise vector as ${\bm \epsilon} = (\epsilon_1,...,\epsilon_m)^{\TT}$.
		Note that 
		\[
		Y_{jj} = \frac{1}{m}\sum\limits_{i=1}^{m}y_ia_{jj}^i = \frac{1}{m}\sum\limits_{i=1}^{m}  (\x^{\TT}\A_i\x) a_{jj}^i+ \frac{1}{m}\sum\limits_{i=1}^{m}\epsilon_ia_{jj}^i.
		\]
		For the first term, 
		a simple calculation gives that the expectation   $\mbbE[Y_{jj}]=x_j^2$. Since $(\x^{\TT}\A_i\x)a_{jj}^i, i=1,\ldots,m$ are independent sub-exponential random variables with \begin{equation*}
			\norm{(\x^{\TT}\A_i\x) a^i_{jj}}_{\psi_1} \le \norm{(\x^{\TT}\A_i\x)}_{\psi_2}\norm{a^i_{jj}}_{\psi_2} \le c_1\norm{\x}^2
		\end{equation*}
		for some constant $c_1\ge 1$, it then follows from Lemma \ref{le:bern} that for any $t_1 \ge 0$,
		\begin{equation*}
			\mbbP\left(\abs{\frac{1}{m}\sum\limits_{i=1}^{m}(\x^{\TT}\A_i\x)a_{jj}^i-x_j^2} > t_1\right) \le 2\exp\left(-C_1\min\left(\frac{t_1^2}{c_1^2 \norm{\x}^4},\frac{t_1}{c_1 \norm{\x}^2}\right)m\right),
		\end{equation*}
		where $C_1>0$ is a universal constant. Setting $\textstyle t_1=c_1\norm{\x}^2\sqrt{\displaystyle\frac{6\log{(mn)}}{C_1m}} $, we have 
		\[
		\mbbP\left(\abs{\frac{1}{m}\sum\limits_{i=1}^{m}(\x^{\TT}\A_i\x)a_{jj}^i-x_j^2} >c_1\norm{\x}^2\sqrt{\displaystyle\frac{6\log{(mn)}}{C_1m}}   \right) \le 2 (mn)^{-6},
		\]	
		provided $m\ge 6C_1^{-1}  \log (mn)$.	 Taking the union bound over $j\in [n]$, we obtain that 
		\[
		\sup_{j\in [n]} \quad \abs{\frac{1}{m}\sum\limits_{i=1}^{m}(\x^{\TT}\A_i\x)a_{jj}^i-x_j^2}  \le C_2 \norm{\x}^2\sqrt{\displaystyle\frac{\log{(mn)}}{m}} 
		\]
		holds with probability at least $1-2(mn)^{-5}$.
		For the second term $\frac{1}{m}\sum\limits_{i=1}^{m}\epsilon_ia_{jj}^i$, conditionally on the $\epsilon_i$'s, Lemma \ref{le:Hoeff} gives that 
		\[
		\mbbP\left(\abs{\sum\limits_{i=1}^{m}\epsilon_ia_{jj}^i } > C_3\norm{\bm\epsilon}\sqrt{\log(mn)}\right) \le 2(mn)^{-6}.
		\]
		Taking the union bound over $j\in [n]$, we obtain that
		\[
		\max_{1\le j\le n} \abs{\sum\limits_{i=1}^{m}\epsilon_ia_{jj}^i } \le  C_3 \norm{\bm \epsilon}\sqrt{\log (mn)} \le C_4 \sigma \sqrt{m\log (mn)}
		\]
		holds with probability at least $1-2n^{-5}$,  where we use \cite[Lemma A.7]{cai2016optimal}  in the last inequality, and $C_3, C_4>0$ are universal constants.
		Combining  the two bounds, we obtain that with probability at least $1-4n^{-5}$, it holds
		\begin{equation*}
			\abs{{Y}_{jj}-x_j^2} \le C \left(1 +\frac{\sigma}{\norm{\x}^2}\right) \norm{\x}^2 \sqrt{\frac{\log{(mn)}}{m}} 
		\end{equation*}
		for all $j=1,\ldots,n$. Here, $C>0$ is a universal constant.
		For all $j \in \mcalS^c$, noting that $\mbbE[Y_{jj}]=0$, it gives
		\begin{equation*}
			Y_{jj}  \le C  \left(1 +\frac{\sigma}{\norm{\x}^2}\right) \norm{\x}^2 \sqrt{\frac{\log{(mn)}}{m}}  \quad \mbox{for all} \quad j\in  \mcalS^c.
		\end{equation*} 
		This complete the proof.
	\end{proof}

	Now, we are ready to prove Theorem \ref{Theorem1}.  
	First, we establish that the true signal $\x$ is well-approximated by its projection $\x_{\widehat{\mathcal{S}}}$ onto the estimated support set $\widehat{\mathcal{S}}$. Second, we demonstrate that this projection $\x_{\widehat{\mathcal{S}}}$ closely approximates our initial estimate $\x^0$. The triangle inequality then yields the desired result.

	\begin{lemma} \label{le:normshat}
		Let $\widehat{\mcalS}$ be the set given in Algorithm \ref{algorithm1:Initialization}. For any $\delta_1 > 0$,
		with probability at least $1-4n^{-5}$,  it holds
		\begin{equation*}
			\norm{\x-\x_{\hat{\mcalS}}} \le \delta_1\norm{\x}
		\end{equation*}
		when $m \ge C(\delta_1)\left(1 +\frac{\sigma}{\norm{\x}^2}\right)^2 s^2 \log (mn)$.
		Here, $C(\delta_1) >0$ is a constant depends only on  $\delta_1$.
		\label{lemma4}
	\end{lemma}
	\begin{proof} 
		Define 
		\begin{equation*}
			\mcalS_+=\left\{j \in \mcalS \mid x_j^2> 2C \left(1 +\frac{\sigma}{\norm{\x}^2}\right) \norm{\x}^2 \sqrt{\frac{\log{(mn)}}{m}} \right\},
		\end{equation*}
		where the constant $C>0$ corresponds exactly to the one in \eqref{eq:bound1}. We first show that with high probability it holds $\mcalS_+ \subset \widehat{\mcalS}$. To this end, for any $ j \in \mcalS_+$, it follows from Lemma \ref{lemma2} that with probability at least $1-4n^{-5}$, one has
		\begin{equation*}
			Y_{jj} \ge x_j^2-\abs{Y_{jj}-x_j^2} \ge C\left(1 +\frac{\sigma}{\norm{\x}^2}\right) \norm{\x}^2 \sqrt{\frac{\log{(mn)}}{m}}.
		\end{equation*}
		And for any $j \in \mcalS^c$, it holds
		\[
		Y_{jj} \le C\left(1 +\frac{\sigma}{\norm{\x}^2}\right) \norm{\x}^2 \sqrt{\frac{\log{(mn)}}{m}}. 
		\]
		Since $\widehat{\mcalS}$ contains the indices of the maximum $s$ elements of $Y_{jj}$, it gives $\mcalS_+ \subset \widehat{\mcalS} \subset \mcalS$. Therefore, we have
		\begin{eqnarray*}
			\norm{\x-\x_{_{\widehat{\mcalS}}}}^2  \le  \norm{\x- \x_{\mcalS_+ }}^2 \le  s\cdot 2C\left(1 +\frac{\sigma}{\norm{\x}^2}\right)\norm{\x}^2\sqrt{\frac{\log{(mn)}}{m}}  \le  \delta_1^2\norm{\x}^2,
		\end{eqnarray*}	
		provided that $m\ge 4C^2 \delta_1^{-4}\left(1 +\frac{\sigma}{\norm{\x}^2}\right)^2 s^2 \log (mn)$. This completes the proof.
	\end{proof}

	\begin{lemma} \label{le:5}
		\label{lemma5}
		For any fixed $0<\delta<1$,  then with probability at least $1-4n^{-5}-4\exp(-c\delta^2 m)-2\exp(-c's)$, we have
		\begin{equation*}
			{\rm dist}\left(\x^0,\x_{_{\widehat{\mcalS}}}\right) := \min\left(\norm{\x^0-\x_{_{\widehat{\mcalS}}}},\norm{\x^0+\x_{_{\widehat{\mcalS}}}}\right)  \le \delta \norm{\x},
		\end{equation*} 
		provided that $m \ge C(\delta)(1+\frac{\sigma}{\norm{x}^2})^2s^2\log(mn)$. Here, $c, c'$ are universal positive constant and $C(\delta)$ is a constant depends on $\delta$.
	\end{lemma}
	\begin{proof} 
		A simple calculation gives
		\begin{equation*}
			\mbbE\left[\Y\right]=\mbbE\left[\frac{1}{m}\sum\limits_{i=1}^{m}y_i\A_i\right]=\x\x^{\TT}.
		\end{equation*}
		This implies that $\mbbE\left[\Y_{\widehat{\mcalS}}\right] =\x_{_{\widehat{\mcalS}}}\x_{_{\widehat{\mcalS}}}^{\TT}$.
		We claim that for any $0<\delta_2<1$, when $m\ge C_0\delta_2^{-2} s \log(en/s)$ with large positive constant $C_0$, with probability at least $1-2\exp(-c_3 \delta_2^2 m)-2\exp(-c's)$,
		\begin{equation} \label{eq:cla1}
			\norm{\Y_{\mcalS'}-\mbbE\left[\Y_{\mcalS'}\right]} \le \delta_2 \norm{\x}^2
		\end{equation}
		holds for all $\mcalS' \subset [n]$ with $|\mcalS'| \le s$.  Note that $|\widehat{\mcalS}|=s$. It implies that 
		\begin{equation} \label{eq:conYsp}
			\norm{\Y_{\widehat{\mcalS}}-\mbbE\left[\Y_{\widehat{\mcalS}}\right]} \le \delta_2 \norm{\x}^2. 
		\end{equation}
		Since $\x^0$ is the leading left singular vector of $\Y_{\widehat{\mcalS}}$, a variant of Wedin's sin$\Theta$ theorem \cite[Theorem 2.1]{Dopico} gives
		\begin{equation} \label{eq:xno1}
			\min_{\alpha=\pm 1}  \norm{\frac{\x^0}{\norm{\x^0}}- \alpha \frac{\x_{\widehat{\mcalS}}}{\norm{\x_{\widehat{\mcalS}}}}}   \le \frac{c_1 \norm{\Y_{\widehat{\mcalS}}-\mbbE\left[\Y_{\widehat{\mcalS}}\right]} }{\sigma_1(\mbbE\left[\Y_{\widehat{\mcalS}}\right])-\sigma_2(\Y_{\widehat{\mcalS}})} \le \frac{c_1 \delta_2 \norm{\x}^2}{\norm{\x}^2- \delta_2 \norm{\x}^2}=\frac{c_1 \delta_2}{1-\delta_2}.
		\end{equation}
		where we use Lemma \ref{le: Weyl} in the second inequality.
		When the variance of the noise is much smaller compared to the magnitude of the measurement, Lemma \ref{le:normphi} still remains valid.
		According to Lemma \ref{le:normphi} and Lemma \ref{le:normshat}, for any $0<\delta_1<1$,  with probability at least $1-4n^{-5}-2\exp(-c\delta_1^2 m)$ we have
		\begin{equation} \label{eq:xno2}
			(1-\delta_1) \norm{\x} \le \max\{\norm{\x^0},\norm{\x_{\widehat{\mcalS}}} \} \le (1+\delta_1)\norm{\x},
		\end{equation}
		provided that $m \ge C(\delta_1)(1+\frac{\sigma}{\norm{x}^2})^2s^2\log{(mn)}$.
		Combining \eqref{eq:xno1} and \eqref{eq:xno2}, we can easily obtain that 
		\begin{eqnarray*}
			\min_{\alpha=\pm 1}  \norm{\x^0- \alpha \x_{_{\widehat{\mcalS}}}} &\le& \min_{\alpha=\pm 1}  \norm{\frac{\x^0}{\norm{\x^0}}- \alpha \frac{\x_{\widehat{\mcalS}}}{\norm{\x_{\widehat{\mcalS}}}}} \norm{\x} +  \abs{\norm{\x^0}-\norm{\x}} + \abs{\norm{\x_{\widehat{\mcalS}}}-\norm{\x}} \\
			& \le &\frac{c_1 \delta_2}{1-\delta_2}\norm{\x}+2\delta_1\norm{\x} \le \delta \norm{\x}
		\end{eqnarray*}
		by letting $\delta_1$ and $\delta_2$ be small enough. This gives the conclusion.
		
		It remains to prove the claim \eqref{eq:cla1}. By the triangle inequality, we obtain
		\begin{equation*}
			\begin{aligned}
				&\norm{\frac{1}{m}\sum\limits_{i=1}^{m}y_i[\A_i]_{\mcalS',\mcalS'}-\x_{_{\mcalS'}}\x_{_{\mcalS'}}^{\TT}}\\ &\le \norm{\frac{1}{m}\sum\limits_{i=1}^{m}\x^{\TT}\A_i\x[\A_i]_{\mcalS',\mcalS'}-\x_{_{\mcalS'}}\x_{_{\mcalS'}}^{\TT}}+\norm{\frac{1}{m}\sum\limits_{i=1}^{m}\epsilon_i[\A_i]_{\mcalS',\mcalS'}}.
			\end{aligned}
		\end{equation*}
		For any fixed $\mcalS' \subset [n]$ with $|\mcalS'| \le s$. Let $\mathcal{N}_\epsilon$ be a $\frac{1}{4}$-net of the unit sphere $\mathbb S^{s-1}$. For the first term, we have
		\begin{equation*}
			\begin{aligned}
				&\norm{\frac{1}{m}\sum\limits_{i=1}^{m}\x^{\TT}\A_i\x[\A_i]_{\mcalS',\mcalS'}-\x_{_{\mcalS'}}\x_{_{\mcalS'}}^{\TT}} \\
				&\le 2\sup\limits_{\u,\v \in \mcalN_{\1/4}}\left|\frac{1}{m}\sum\limits_{i=1}^{m}\x^{\TT}\A_i\x\u^{\TT}[\A_i]_{\mcalS',\mcalS'}\v-(\u^{\TT}\x_{_{\mcalS'}})(\v^{\TT}\x_{_{\mcalS'}})\right|. 
			\end{aligned}
		\end{equation*}
		Since $\textstyle \x^{\TT}\A_i\x\u^{\TT}[\A_i]_{\mcalS',\mcalS'}\v, \ i=1,...,m$, are independent sub-exponential variables, Lemma \ref{le:bern} gives
		\begin{eqnarray*}
			\mbbP&\left(\left|\frac{1}{m}\sum\limits_{i=1}^{m}\x^{\TT}\A_i\x\u^{\TT}[\A_i]_{\mcalS',\mcalS'}\v-(\u^{\TT}\x_{_{\mcalS'}})(\v^{\TT}\x_{_{\mcalS'}})\right|\ge t_2\right)\\
			& \le 2\exp\left(-C_1 \min\left(\frac{t_2^2}{K_2^2},\frac{t_2}{K_2}\right)m\right),
		\end{eqnarray*}
		where $\textstyle K_2=\max\limits_{i}\norm{\x^{\TT}\A_i\x\u^{\TT}[\A_i]_{\mcalS',\mcalS'}\v-(\u^{\TT}\x_{_{\mcalS'}})(\v^{\TT}\x_{_{\mcalS'}})}_{\psi_1} $ obeys
		\begin{eqnarray*}
			K_2 &\le& c_1\max\limits_{i}\norm{\x^{\TT}\A_i\x\u^{\TT}[\A_i]_{\mcalS',\mcalS'}\v}_{\psi_{1}}\\
			&\le& c_2\norm{\x^{\TT}\A_i\x}_{\psi_{2}}\norm{\u^{\TT}[\A_i]_{\mcalS',\mcalS'}\v}_{\psi_{2}}\le c_2\norm{\x}^2.
		\end{eqnarray*}
		Setting $\textstyle t_2=\frac{\delta_3}{2} \norm{\x}^2$, taking the union bound over $\mathcal{N}_\epsilon$, we obtain that with probability at least $1-2\exp(-c_0 \delta_3^2 m)$ with a universal constant $c_0$, it holds
		\[
		\norm{\frac{1}{m}\sum\limits_{i=1}^{m}\x^{\TT}\A_i\x[\A_i]_{\mcalS',\mcalS'}-\x_{_{\mcalS'}}\x_{_{\mcalS'}}^{\TT}} \le \delta_3 \norm{\x}^2,
		\]
		provided $m\ge C_0\delta_3^{-2} s$ for some large constant $C_0>0$. Note that $\sum\limits_{k=1}^{s} \left(\begin{array}{c} n\\k\end{array} \right) \le \left(en/s\right)^{s}$. Taking the union bound over all $\mcalS'$, we obtain the conclusion that 
		\[
		\norm{\frac{1}{m}\sum\limits_{i=1}^{m}\x^{\TT}\A_i\x[\A_i]_{\mcalS',\mcalS'}-\x_{_{\mcalS'}}\x_{_{\mcalS'}}^{\TT}} \le \delta_3 \norm{\x}^2 \qquad \mbox{for all} \quad |\mcalS' |\le s
		\]
		holds with probability at least 
		\[
		1- \left(en/s\right)^{s}\cdot 2\exp(-c_0 \delta_3^2 m) \ge 1-2\exp(-c_3 \delta^2 m),
		\]
		provided $m\ge C_0\delta^{-2} s \log(en/s)$ with any fixed $0<\delta<1$. Combining with Lemma \ref{le:epAi}, we have 
		\begin{equation*}
			\norm{\frac{1}{m}\sum\limits_{i=1}^{m} \epsilon_i  [\A_i]_{\mcalS,\mcalS} } \le C_0 \sigma \sqrt{ \frac{s \log(en/s)}m} \le \delta_3\norm{\x}^2. 
		\end{equation*}
		We complete the proof.
	\end{proof}

	Now, we are ready to prove Theorem \ref{Theorem1}.
	
	\begin{proof}[Proof of Theorem \ref{Theorem1}]
		
		According to Lemma \ref{le:normshat}, with probability at least $1-4n^{-5}$, when $m \ge C(\delta_0)\left(1 +\frac{\sigma}{\norm{\x}^2}\right)^2 s^2 \log (mn)$,  we have
		\begin{equation*}
			\norm{\x-\x_{\hat{\mcalS}}} \le \frac{\delta_0}2\norm{\x}.
		\end{equation*}
		Lemma \ref{le:5} gives that  when $m \ge C(\delta_0)(1+\frac{\sigma}{\norm{\x}^2})^2s^2\log{(mn)}$, with probability at least $1-4n^{-5}-4\exp(-c\delta_0^2 m)-2\exp(-c's)$, it holds
		\begin{equation*}
			\min\left(\norm{\x^0-\x_{_{\widehat{\mcalS}}}},\norm{\x^0+\x_{_{\widehat{\mcalS}}}}\right)  \le \frac{\delta_0}2\norm{\x}.
		\end{equation*} 
		Combining the above estimators and applying the triangle inequality, we immediately obtain that,  with  probability at least $1-8n^{-5}-4\exp(-c\delta_0^2 m)-2\exp(-c's)$, 
		\begin{equation*}
			{\rm dist}\left(\x^0,\x\right)=\min\left(\norm{\x^0-\x},\norm{\x^0+\x}\right) \le \delta_0\norm{\x}
		\end{equation*}
		holds, provided $m \ge C(\delta_0)\left(1 +\frac{\sigma}{\norm{\x}^2}\right)^2 s^2 \log (mn)$. This completes the proof.	
	\end{proof}

	\subsection{Proof of Theorem \ref{Theorem3}} \label{sec:Thnoisy}
	
	Our convergence analysis follows from the strategy used in \cite{cai2023fast, Dai}, with key modifications to address the sparse quadratic setting. We begin by establishing several auxiliary lemmas.
	
	First, for any index subset $\mcalS \subseteq [n]$, we define
	\begin{equation*}
		\mcalX_{\mcalS} = \{\u \in \mbbS^{n-1}:{\rm supp}(\u) \subset \mcalS \},
	\end{equation*} 
	where $\mathbb S^{n-1}:=\{\z \in \mathbb R^n : \norm{\z} = 1\}$. The following lemma extends \cite[Lemma 12]{chen2025solving}. Whereas the original result applied only to a fixed subset $\mcalS \subseteq [n]$, our result establishes uniform guarantees that hold simultaneously for all $\mcalS \subseteq [n]$.

	%	\begin{lemma}
		%		Given any $\textstyle \p,\q \in \mbbR^n$, $\A_i \sim \mcalN^{n\times n}(0,1)$, we have the following expectations:
		%		\begin{equation}
			%		\mbbE\left[\A_i\p\q^{\TT}\A_i\right]=\q\p^{\TT},
			%	    \end{equation}
		%	    \begin{equation}
			%		\mbbE\left[\A_i\p\q^{\TT}\A_i^{\TT}\right]=(\p^{\TT}\q)\I_{n \times n}.
			%		\end{equation}
		%		And for the first term, it has one singular value of $(\p^{\TT}\q)$ and all other singular value are $0$. For the second term, all of its singular values are $(\p^{\TT}\q)$. 
		%		\label{lemma6}
		%	\end{lemma}
	%	\begin{proof} 
		%	Through a simple calculation, the first and second relations can be directly established. By the definition of singular value, the remainder of the lemma holds.
		%	\end{proof}

	\begin{lemma} \label{lemma7}
		Assume that $\A_i \sim \mcalN^{n\times n}(0,1)$. For any integer $0<s <n$, if $m\ge C(\delta)s\log{n}$, then with probability at least $\textstyle 1-2\exp(-\gamma_1m)$, it holds
		\begin{equation}
			\sup\limits_{\mcalS \subseteq [n] \atop \abs{\mcalS} \le 4s}	\sup\limits_{\p,\q \in \mcalX_{\mcalS} }\norm{\frac{1}{m}\sum\limits_{i=1}^m [\A_i]_{\mcalS, \mcalS}\p_\mcalS\q_\mcalS^{\TT}[\A_i]_{\mcalS, \mcalS}-\q_\mcalS\p_\mcalS^{\TT}} \le \delta \norm{\p_\mcalS}\norm{\q_\mcalS}
			\label{eq:concentration1}
		\end{equation}
		and
		\begin{equation}
			\sup\limits_{\mcalS \subseteq [n] \atop \abs{\mcalS} \le 4s} \sup\limits_{\p,\q \in \mcalX_{\mcalS} }\norm{\frac{1}{m}\sum\limits_{i=1}^m[\A_i]_{\mcalS, \mcalS}\p_\mcalS\q_\mcalS^{\TT}[\A_i^{\TT}]_{\mcalS, \mcalS}-(\p_\mcalS^{\TT}\q_\mcalS)I_{\mcalS,\mcalS}} \le \delta \norm{\p_\mcalS}\norm{\q_\mcalS}.
			\label{eq:concentration2}
		\end{equation}
		Here $C(\delta)$ is a constant depending on $\delta$, and $\gamma_1>0$ is a universal constant. 
	\end{lemma}
	\begin{proof}
		The proof is similar to that of \cite[Lemma 12]{chen2025solving}, so we only give a brief description.  For any fixed $\mcalS \subseteq [n]$ with $\abs{\mcalS} \le  4s$,  let $\mathcal{N}_{1/4}$ be a $\frac{1}{4}$-net of $\mcalX_{\mcalS}$ with cardinality $\abs{\mcalN_{1/4}} \le 9^{4s}$.  Then for any fixed $\p,\q \in \mcalX_{\mcalS}$,  it holds
		\begin{equation}
			\begin{aligned}
				&\norm{\frac{1}{m}\sum\limits_{i=1}^m [\A_i]_{\mcalS, \mcalS}\p_\mcalS\q_\mcalS^{\TT}[\A_i]_{\mcalS,\mcalS}-\q_\mcalS\p_\mcalS^{\TT}} \\
				&\le 2\sup\limits_{\u,\v \in \mathcal{N}_{1/4}}\abs{\frac{1}{m}\sum\limits_{i=1}^m \u^{\TT}[\A_i]_{\mcalS, \mcalS}\p_\mcalS\q_\mcalS^{\TT}[\A_i]_{\mcalS,\mcalS}\v-\u^{\TT}\q_\mcalS\p_\mcalS^{\TT}\v}.
			\end{aligned}
			\label{eq:normineq1}
		\end{equation}
		It is easy to see that $\textstyle \mbbE\left[\u^{\TT}[\A_i]_{\mcalS,\mcalS}\p_\mcalS\q_\mcalS^{\TT}[\A_i]_{\mcalS,\mcalS}\v\right]=(\u^{\TT}\q_\mcalS)(\p_{\mcalS}^{\TT}\v)$.  Therefore, for any $\textstyle (\u,\v)\in \mathcal{N}_{1/4}$ and $t_3>0$, by Lemma \ref{le:bern}, we obtain 
		\begin{eqnarray*}
			\mbbP&\left(\abs{\frac{1}{m}\sum\limits_{i=1}^m \u^{\TT}[\A_i]_{\mcalS, \mcalS}\p_\mcalS\q_\mcalS^{\TT}[\A_i]_{\mcalS,\mcalS}\v-\u^{\TT}\q_\mcalS\p_\mcalS^{\TT}\v} > t_3\right) \\
			&\le 2\exp\left(-C_1\min\left(\frac{t_3^2}{K_3^2},\frac{t_3}{K_3}\right)m\right),
		\end{eqnarray*}
		where $\textstyle K_3=\max\limits_i\norm{\u^{\TT}[\A_i]_{\mcalS,\mcalS}\p_\mcalS\q_\mcalS^{\TT}[\A_i]_{\mcalS,\mcalS}\v-(\u^{\TT}\q_\mcalS)(\p_{\mcalS}^{\TT}\v)}_{\psi_1} \le  c_2\norm{\p_\mcalS}\norm{\q_\mcalS} $.  Taking the union bound over $\u, \v\in \mcalN_{1/4}$, we immediately have 
		%\begin{equation*}
		%	\begin{aligned}
			%		&\mbbP\left(\sup\limits_{\u,\v \in \mathcal{N}_{1/4}}\abs{\frac{1}{m}\sum\limits_{i=1}^m \u^{\TT}\A^i_{\mcalS, \mcalS}\p_\mcalS\q_\mcalS^{\TT}\A^i_{\mcalS,\mcalS}\v-\u^{\TT}\q_\mcalS\p_\mcalS^{\TT}\v} > t_3\right)\\
			%		&\le 2\exp\left(2s\log{9}-C_1\min\left(\frac{t_3^2}{K_3^2},\frac{t_3}{K_3}\right)m\right),
			%	\end{aligned}
		%\end{equation*}
		%Then we combine this with (\ref{eq:normineq1}) to obtain
		\begin{equation} \label{eq:fixqp}
			\begin{aligned}
				\mbbP&\left(\norm{\frac{1}{m}\sum\limits_{i=1}^m [\A_i]_{\mcalS, \mcalS}\p_\mcalS\q_\mcalS^{\TT}[\A_i]_{\mcalS,\mcalS}-\q_\mcalS\p_\mcalS^{\TT}} > 2t_3\right) \\
				&\le 2\exp\left(8s\log{9}-C_1\min\left(\frac{t_3^2}{K_3^2},\frac{t_3}{K_3}\right)m\right)
			\end{aligned}
		\end{equation}
		for any fixed $\p,\q \in \mcalX_{\mcalS}$.  
		
		To establish the uniform bound of \eqref{eq:fixqp} for all $\p,\q \in \mcalX_{\mcalS}$, we define 
		\begin{equation}
			\G(\p_\mcalS,\q_\mcalS):=\frac{1}{m}\sum\limits_{i=1}^m [\A_i]_{\mcalS, \mcalS}\p_\mcalS\q_\mcalS^{\TT}[\A_i]_{\mcalS,\mcalS}-\q_\mcalS\p_\mcalS^{\TT}.
		\end{equation}
		According to \eqref{eq:fixqp},  for any $t_3>0$, taking the union bound over $\p,\q \in \mcalN_{1/4}$ yields  
		\begin{equation}
			\mbbP\left(\sup\limits_{\p,\q \in \mcalN_{1/4}}\norm{\G(\p_\mcalS,\q_\mcalS)} > 2t_3\right) \le 2\exp\left(16s\log{9}-C_1\min\left(\frac{t_3^2}{K_3^2},\frac{t_3}{K_3}\right)m\right).
			\label{eq: bernstein1}
		\end{equation}
		Let  $\textstyle \p_\mcalS^{(0)}, \q_\mcalS^{(0)} \in \mcalX_{\mcalS}$ such that 
		$\textstyle \sup\limits_{\p,\q \in  \mcalX_{\mcalS}}\norm{\G(\p_\mcalS,\q_\mcalS)} = \norm{\G(\p_\mcalS^{(0)},\q_\mcalS^{(0)})}
		$. There exist  $\textstyle \p^{(1)}_\mcalS,\q^{(1)}_\mcalS \in \mcalN_{1/4}$ so that $\textstyle \norm{\p_\mcalS^{(1)}-\p_\mcalS^{(0)}}\le \frac{1}{4}$ and $\textstyle \norm{\q_\mcalS^{(1)}-\q_\mcalS^{(0)}} \le \frac{1}{4}$. This gives
		\begin{eqnarray*}
			\norm{\G(\p_\mcalS^{(0)},\q_\mcalS^{(0)})}&= & \norm{\G(\p_\mcalS^{(0)}-\p_\mcalS^{(1)},\q_\mcalS^{(0)})+\G(\p_\mcalS^{(1)},\q_\mcalS^{(0)}-\q_\mcalS^{(1)})+\G(\p_\mcalS^{(1)},\q_\mcalS^{(1)})}\\
			&\le &  \frac{1}{2}\norm{\G(\p_\mcalS^{(0)},\q_\mcalS^{(0)})} + \sup\limits_{\p,\q \in \mcalN_{1/4}}\norm{\G(\p,\q)},
		\end{eqnarray*}	
		where the first equality due to the bilinearity of $\textstyle\G(\p,\q)$ on $\textstyle(\p,\q)$, the last inequality follows by observing $\textstyle \p_\mcalS^{(0)}-\p_\mcalS^{(1)}=r\p_\mcalS^{(2)}$ for some $\textstyle r \in [0,\frac{1}{4}]$ and $\textstyle \p_\mcalS^{(2)} \in \mcalX_{\mcalS}$ (and similar relation for $\textstyle \q_\mcalS^{(0)}-\q_\mcalS^{(1)}$). Therefore, 
		\begin{equation*}
			\norm{\G(\p_\mcalS^{(0)},\q_\mcalS^{(0)})} \le 2\sup\limits_{\p,\q \in \mcalN_{1/4}}\norm{\G(\p,\q)}.
		\end{equation*}
		Recall that $\textstyle \sup\limits_{\p,\q \in \mcalX_{\mcalS}}\norm{\G(\p_\mcalS,\q_\mcalS)} = \norm{\G(\p_\mcalS^{(0)},\q_\mcalS^{(0)})}$. Combined with (\ref{eq: bernstein1}) we have 
		\begin{equation*}
			\mbbP\left(\sup\limits_{\p,\q \in \mcalX_{\mcalS}}\norm{\G(\p_\mcalS,\q_\mcalS)} > 4t_3\right) \le 2\exp\left(16s\log{9}-C_1\min\left(\frac{t_3^2}{K_3^2},\frac{t_3}{K_3}\right)m\right).
		\end{equation*}
		Noting that $K_3 \le c_2\norm{\p_\mcalS}\norm{\q_\mcalS}$ and  setting $\textstyle t_3 = \frac{\delta}{4}\norm{\p_\mcalS}\norm{\q_\mcalS}$,  we immediately obtain that with probability greater than $\textstyle 1-2\exp(-\gamma_1m)$, 
		\begin{equation*}
			\sup\limits_{\p,\q \in \mcalX_{\mcalS}}\norm{\frac{1}{m}\sum\limits_{i=1}^m [\A_i]_{\mcalS, \mcalS}\p_\mcalS\q_\mcalS^{\TT}[\A_i]_{\mcalS,\mcalS}-\q_\mcalS\p_\mcalS^{\TT}} \le \delta \norm{\p_\mcalS}\norm{\q_\mcalS} 
		\end{equation*} 
		holds for any fixed $\mcalS \subseteq [n]$ with $\abs{\mcalS} \le  4s$, provided that $m > C(\delta) s$ with some large constant $C(\delta)>0$. Finally,  taking a union bound over all possible $\mcalS \subseteq [n]$ with $\abs{\mcalS} \le 4s$ and observing $\sum\limits_{k=1}^{4s} \left(\begin{array}{c} n\\k\end{array} \right) \le \left(en/s\right)^{4s}$,  the desired (\ref{eq:concentration1}) follows when $m > C(\delta)s\log{n}$.   (\ref{eq:concentration2}) can be proved similarly, and we omit it.
	\end{proof}	
	
	\begin{lemma}\cite[Weyl's inequality, eigenvalues]{Weyl}
		\label{le: Weyleigen}
		Let $\M, \N \in \mathbb R^{n\times n}$ be two symmetric matrices. The eigenvalues of $\M$ and $\N$ are denoted as $\lambda_1(\M)\ge \ldots \ge \lambda_n(\M)$ and $\lambda_1(\N)\ge \ldots \ge \lambda_n(\N)$, respectively. Then we have
		\begin{eqnarray*}
			\abs{\lambda_i(\M)-\lambda_i(\N)} \le \norm{\M-\N} \quad \mbox{for} \quad 1 \le i \le n.
		\end{eqnarray*}
		
	\end{lemma}
	
	\begin{lemma} \label{lemma8}
		For any fixed $0<\delta \le 1$, assume that events \eqref{eq:concentration1} and \eqref{eq:concentration2} hold. 
		For any subset $\mcalS, \mcalT \subseteq [n]$ with $\abs{\mcalS} \le 2s$ and $\abs{\mcalT} \le 2s$, and for any vector $\p \in \mbbR^n$ with ${\rm supp}(\p) \subseteq \mcalT$, the following holds:
		\begin{itemize}
			\item[(i)]
			\begin{equation} 	\label{A.4}
				(2-\delta) \norm{\p}^2 \le \norm{\J_\mcalS(\p)^{\TT}\J_\mcalS(\p)} \le  (4+\delta) \norm{\p}^2.		
			\end{equation}
			\item[(ii)] $\textstyle \J_{\mcalS}(\p)^{\TT}\J_{\mcalS}(\p)$ is invertible and 
			\begin{equation} \label{A.6}
				\norm{\left(\J_{\mcalS}(\p)^{\TT}\J_{\mcalS}(\p)\right)^{-1}} \le \frac{1}{\left(2-\delta\right)\norm{\p}^2}.
			\end{equation} 
			\item[(iii)] 
			\begin{equation} \label{eq:sec}
				\norm{\J_\mcalS(\p)^{\TT}\J_{\mcalT \backslash \mcalS}(\p)} \le (1+\delta) \norm{\p}^2.
			\end{equation} 
		\end{itemize}
		Here, the matrix $\J(\p)$ is defined in \eqref{eq:Jx}.
	\end{lemma}
	\begin{proof}
		Observe that 
		\begin{equation*}
			\mbbE\left[\J(\p)^{\TT}\J(\p)\right]=\mbbE\left[\frac{1}{m}\sum\limits_{i=1}^m\widetilde{\A}_i\p\p^{\TT}\widetilde{\A}_i\right]=2\p\p^{\TT}+2\norm{\p}^2\I_{n \times n}.
		\end{equation*}
		Define $\mcalR = \mcalT \cup \mcalS$. Note that $\abs{\mcalR } \le 4s$. It then follows from  \eqref{eq:concentration1} and \eqref{eq:concentration2} that  
		\begin{equation*}
			\norm{\J_\mcalR(\p)^{\TT}\J_\mcalR(\p)-\mbbE\left[\J_\mcalR(\p)^{\TT}\J_\mcalR(\p)\right]} \le \delta \norm{\p}^2.
		\end{equation*}
		Observing that $\J_\mcalS(\p)^{\TT}\J_\mcalS(\p)$ is a submatrix of $\J_\mcalR(\p)^{\TT}\J_\mcalR(\p)$, the interlacing inequality together with Lemma \ref{le: Weyleigen} yields
		\begin{eqnarray}
			\lambda_{\min}\left(\J_\mcalS(\p)^{\TT}\J_\mcalS(\p)\right)& \ge & \lambda_{\min}\left(\J_\mcalR(\p)^{\TT}\J_\mcalR(\p)\right) \notag\\
			&\ge &\lambda_{\min}\left(\mbbE\left[\J_\mcalR(\p)^{\TT}\J_\mcalR(\p)\right]\right)-   \delta \norm{\p}^2 \notag \\
			&\ge & (2-\delta) \norm{\p}^2,  \label{eq:lammin}
		\end{eqnarray}
		where the last inequality comes from the fact that $\lambda_{\min}\left(\mbbE\left[\J_\mcalR(\p)^{\TT}\J_\mcalR(\p)\right]\right)=2\norm{\p}^2$. Similarly, we also have 
		\begin{eqnarray}
			\lambda_{\max}\left(\J_\mcalS(\p)^{\TT}\J_\mcalS(\p)\right)& \le & \lambda_{\max}\left(\J_\mcalR(\p)^{\TT}\J_\mcalR(\p)\right) \notag\\
			&\le &\lambda_{\max}\left(\mbbE\left[\J_\mcalR(\p)^{\TT}\J_\mcalR(\p)\right]\right)+   \delta \norm{\p}^2  \notag\\
			&\le & (4+\delta) \norm{\p}^2.  \label{eq:lammax}
		\end{eqnarray}
		%\begin{equation}
		%\abs{\lambda_{\min}\left(\J_\mcalS(\p)^{\TT}\J_\mcalS(\p)\right)-\sigma_{\min}\left(\mbbE\left[\J_\mcalS(\p)^{\TT}\J_\mcalS(\p)\right]\right)}\le \delta \norm{\p}^2
		%	\label{Weyl1}
		%\end{equation}
		%and 
		%\begin{equation}
		%\abs{\sigma_{\max}\left(\J_\mcalS(\p)^{\TT}\J_\mcalS(\p)\right)-\sigma_{\max}\left(\mbbE\left[\J_\mcalS(\p)^{\TT}\J_\mcalS(\p)\right]\right)} \le \delta \norm{\p}^2.
		%	\label{Weyl2}
		%\end{equation}
		%Note that  $\textstyle \sigma_{\min}\left(\mbbE\left[\J_\mcalS(\p)^{\TT}\J_\mcalS(\p)\right]\right)=2\norm{\p}^2$ and $\textstyle \sigma_{\max}\left(\mbbE\left[\J_\mcalS(\p)^{\TT}\J_\mcalS(\p)\right]\right)=4\norm{\p}^2$. 
		Combining the above two inequalities, we immediately have 
		\[
		(2-\delta) \norm{\p}^2 \le \norm{\J_\mcalS(\p)^{\TT}\J_\mcalS(\p)} \le  (4+\delta) \norm{\p}^2
		\]
		and 
		\[
		\norm{\left(\J_{\mcalS}(\p)^{\TT}\J_{\mcalS}(\p)\right)^{-1}}=\frac{1}{\lambda_{\min}\left(\J_\mcalS(\p)^{\TT}\J_\mcalS(\p)\right)} \le \frac{1}{\left(2-\delta\right)\norm{\p}^2}.
		\]
		For the item (iii),  since $\textstyle\J_\mcalS(\p)^{\TT}\J_{\mcalT \backslash \mcalS}(\p)$ is a submatrix of $\textstyle \J_\mcalR(\p)^{\TT}\J_{\mcalR}(\p)-3\norm{\p}^2I$, it implies that 
		\begin{eqnarray*}
			\norm{\J_\mcalS(\p)^{\TT}\J_{\mcalT \backslash \mcalS}(\p)}  &\le& \norm{\J_\mcalR(\p)^{\TT}\J_{\mcalR}(\p)-3\norm{\p}^2I}\\
			&\le &  \max\{ 4+\delta-3, 3-(2-\delta)  \} \norm{\p}^2  \\
			&=& (1+\delta) \norm{\p}^2,
		\end{eqnarray*}
		where the second inequality comes from \eqref{eq:lammin} and \eqref{eq:lammax}. This completes the proof.
	\end{proof}	
	
	%\begin{lemma}
	%	For any subsets $\mcalS, \mcalT \subseteq [n]$ with $\abs{\mcalS} \le s$ and $\abs{\mcalT} \le s$ for some integer $s <n$. Let $\p \in \mbbR^n$ be a sparse signal with ${\rm supp}(\p) \subseteq \mcalT$. Under events \eqref{eq:concentration1} and \eqref{eq:concentration2}, it holds
	%	\begin{equation}
		%		\norm{\J_\mcalS(\p)^{\TT}\J_{\mcalT \backslash \mcalS}(\p)} \le (1+\delta) \norm{\p}^2.
		%	\end{equation} 
	%	\label{lemma10}
	%\end{lemma}
	%\begin{proof}
	%	Define $\textstyle\mcalR = \mcalT \cup \mcalS$. Since $\textstyle\J_\mcalS(\p)^{\TT}\J_{\mcalT \backslash \mcalS}(\p)$ is a submatrix of $\textstyle \J_\mcalR(\p)^{\TT}\J_{\mcalR}(\p)-3\norm{\p}^2I$, it implies that 
	%	\begin{eqnarray*}
		%		\norm{\J_\mcalS(\p)^{\TT}\J_{\mcalT \backslash \mcalS}(\p)}  &\le& \norm{\J_\mcalR(\p)^{\TT}\J_{\mcalR}(\p)-3\norm{\p}^2I}\\
		%		&\le &  \max\{ 4+\delta-3, 3-(2-\delta)  \} \norm{\p}^2  \\
		%	  &=& (1+\delta) \norm{\p}^2,
		%	\end{eqnarray*}
	%	where the second inequality comes from \eqref{A.4}. This completes the proof.
	%\end{proof}	

	\begin{lemma} \label{lemma9}
		For any fixed $0<\delta \le 0.01$, assume that events \eqref{eq:concentration1} and \eqref{eq:concentration2} hold.
		Let $\x$ be the target $s$-sparse signal. 
		Given an s-sparse estimate $\x^k$ satisfying $\norm{\x^k-\x} \le \delta \norm{\x}$. Define the vector obtained by one iteration of IHT with stepsize $\mu^k$ to be 
		\begin{equation*}
			\u^k = \mathcal{H}_s\left(\x^k-\mu^k\nabla f(\x^k)\right).
		\end{equation*}
		If the step size 	$\mu^k \in (\frac{\underline\mu}{\norm{\x}^2}, \frac{\overline\mu}{\norm{\x}^2})$ for some universal constants  $\underline\mu, \overline\mu \in (0,1)$, it holds
		\begin{equation*}
			\norm{\u^k-\x} \le \rho \norm{\x^k-\x}+ \frac{\sqrt5+1}{2} \cdot\mu^k \cdot C_0 \sigma (1+\delta)\norm{\x}  \sqrt{ \frac{s \log(en/s)}m},
		\end{equation*}
		where $\rho\in (0,0.95)$ is a constant depending only on $\delta$,  $\underline\mu$ and $\overline\mu$.
	\end{lemma}

	\begin{proof} 
		Define $\mathcal{T}_{k}={\rm supp}(\x) \cup {\rm supp}(\x^k)$, $\mathcal{T}_{k+1}={\rm supp}(\x) \cup  {\rm supp}(\u^k)$, and $\textstyle \v^k=\x^k-\mu^k\nabla f(\x^k)$. Since $\u^k$ is the best $s$-term approximation of $\v^k$, it implies $\u^k$ is also the best $s$-term approximation of $\v^k_{\mathcal{T}_{k+1}}$. Combining with the fact that  $\x_{\mathcal{T}_{k+1}}$ is an $s$-sparse vector,  it then follows \cite[Thoerem 1]{shenjie}  that  
		\begin{equation} \label{eq:2vxt}
			\norm{\u^k-\x}  \le  \frac{\sqrt5+1}{2}\norm{\v^k_{\mathcal{T}_{k+1}}-\x_{\mathcal{T}_{k+1}}}.
		\end{equation}
		Therefore, it suffices to upper bound $\norm{\v^k_{\mathcal{T}_{k+1}}-\x_{\mathcal{T}_{k+1}}}$.  Using the definition of $\v^k$, we have
		\begin{equation}
			\norm{\v^k_{\mathcal{T}_{k+1}}-\x_{\mathcal{T}_{k+1}}}=\norm{\x_{\mathcal{T}_{k+1}}^k-\mu^k\nabla_{\mathcal{T}_{k+1}}f(\x^k)-\x_{\mathcal{T}_{k+1}}}.
			\label{eq1}
		\end{equation}
		In noisy case, by (\ref{eq:gradf}), we have 
		\begin{equation*}
			\nabla f(\x)=\frac{1}{m}\sum\limits_{i=1}^{m}(\x^{\TT}\A_i\x-\x^{\TT}\A_i\x-\epsilon_i)\widetilde{\A}_i\x :=\nabla f_{\mbox{clean}}(\x)- \frac{1}{m}\sum\limits_{i=1}^{m} \epsilon_i \widetilde{\A}_i\x.
		\end{equation*}
		By the definitions of $\J(\x)$ given in \eqref{eq:Jx}, a simple calculation gives
		\begin{eqnarray} \label{eq2}
			\nabla_{\mathcal{T}_{k+1}}f_{\mbox{clean}}(\x^k) &= & \frac{1}{m}\sum\limits_{i=1}^{m}\left(\left(\x^k\right)^{\TT}\A_i \x^k-\x^{\TT}\A_i \x \right)[\widetilde{\A}_i]_{\mcalT_{k+1,:}}\x^k \notag \\
			&= & \frac{1}{2}\J_{\mathcal{T}_{k+1}}\left(\x^k\right)^{\TT}\J\left(\x^k + \x\right) \left(\x^k-\x\right) \notag \\
			&= & \left(\J_{\mathcal{T}_{k+1}}\left(\x^k\right)^{\TT}\J\left(\x^k\right)-\frac{1}{2}\J_{\mathcal{T}_{k+1}}\left(\x^k\right)^{\TT}\J\left(\x^k-\x\right)\right) \cdot \left(\x^k-\x\right),
		\end{eqnarray}
		where the second equality comes from the fact that 
		\[
		\left(\x^k\right)^{\TT}\A_i \x^k-\x^{\TT}\A_i \x=\frac12\cdot (\x^k+\x)^\TT \widetilde{\A}_i (\x^k-\x),
		\]
		and the third equality comes from the linearity of $\J(\x^k)$. For simplicity, we denote $\h^k=\x^k-\x$.  Substituting $\nabla_{\mathcal{T}_{k+1}}f(\x^k)$ into (\ref{eq1}), one has
		\begin{eqnarray*}
			&&\norm{\v^k_{\mathcal{T}_{k+1}}-\x_{\mathcal{T}_{k+1}}} \\
			&= & \Big\|\left(\I-\mu^k\J_{\mathcal{T}_{k+1}}\left(\x^k\right)^{\TT}\J_{\mathcal{T}_{k+1}}\left(\x^k\right)\right)\h^k_{\mcalT_{k+1}} +\frac{\mu^k}{2}\J_{\mathcal{T}_{k+1}}\left(\x^k\right)^{\TT}\J_{\mathcal{T}_k}(\h^k)\\
			&&\cdot\h^k_{\mcalT_k}+\mu^k\J_{\mathcal{T}_{k+1}}(\x^k)^{\TT}\J_{\mathcal{T}_k \backslash \mathcal{T}_{k+1}}(\x^k)\h^k_{\mcalT_k\backslash \mcalT_{k+1}}+ \frac{\mu^k}{m}\sum\limits_{i=1}^{m} \epsilon_i [\widetilde{\A}_i]_{\mathcal{T}_{k+1},:}\x^k\Big\|\\
			&\le&  \underbrace{\norm{\left(\I-\mu^k\J_{\mathcal{T}_{k+1}}\left(\x^k\right)^{\TT}\J_{\mathcal{T}_{k+1}}\left(\x^k\right)\right)\h^k_{\mcalT_{k+1}}}}_{I_1}  + \frac{\mu^k}{2} \underbrace{\norm{\J_{\mathcal{T}_{k+1}}\left(\x^k\right)^{\TT}\J_{\mathcal{T}_k}\left(\h^k\right)\h^k_{\mcalT_k}}}_{I_2}\\
			&&+\mu_k \underbrace{\norm{\J_{\mathcal{T}_{k+1}}(\x^k)^{\TT}\J_{\mathcal{T}_k \backslash \mathcal{T}_{k+1}}(\x^k)\h^k_{\mcalT_k\backslash \mcalT_{k+1}}}}_{I_3}+\mu^k \underbrace{\cdot C_0 \sigma \norm{\x^k}  \sqrt{ \frac{s \log(en/s)}m}}_{I_4}.
		\end{eqnarray*}
		where the inequality comes from the triangle inequality and Lemma \ref{le:epAi} with $\mcalS:=\mathcal{T}_{k+1} \cup \mbox{supp}(\x^k)$. 
		We will estimate $I_1, I_2$ and $I_3$ sequentially.  
		For $I_1$, by the assumption $\norm{\x^k-\x} \le \delta \norm{\x}$, we have
		\begin{equation}
			\label{boundxk}
			\left(1-\delta \right)\norm{\x} \le \norm{\x^k} \le \left(1+\delta\right)\norm{\x}.
		\end{equation} 	
		Applying Lemma \ref{lemma8} with  $\mcalS = \mcalT_{k+1}$, we obtain 
		\begin{eqnarray*} \label{eq:bound xk}
			\norm{\I-\mu^k\J_{\mathcal{T}_{k+1}}(\x^k)^{\TT}\J_{\mathcal{T}_{k+1}}(\x^k)}  & \le  & \max\{ 1-\mu^k(2-\delta) \norm{\x^k}^2,  \mu^k(4+\delta) \norm{\x^k}^2 -1 \}.
		\end{eqnarray*}
		%provided 
		%\begin{equation} \label{eq:mu2}
		% \mu^k \in \left( \frac{1}{3(1-\delta)^2 \norm{\x}^2},  \frac{2}{(4+\delta)(1+\delta)^2 \norm{\x}^2} \right).
		% \end{equation}
	This implies 
	\begin{equation*}
		I_1  \le \norm{\I-\mu^k\J_{\mathcal{T}_{k+1}}(\x^k)^{\TT}\J_{\mathcal{T}_{k+1}}(\x^k)} \| \h^k_{\mcalT_{k+1}}\| \le \max\{1-\mu^k\zeta_1, \mu^k\zeta_2-1 \} \| \h^k_{\mcalT_{k+1}}\|,
	\end{equation*}
	where $\zeta_1:=(2-\delta)(1-\delta)^2\norm{\x}^2$ and  $\zeta_2:=(4+\delta)(1+\delta)^2\norm{\x}^2$. 
	
	For the second term $I_2$, applying Lemma \ref{lemma8} with   $\mcalS=\mcalT_k$ and by (\ref{boundxk}), we have 
	\begin{equation*}
		\norm{ \J_{\mathcal{T}_k}(\h^k)} \le  \sqrt{\left(4+\delta\right)} \norm{\h^k} \le \sqrt{\left(4+\delta\right)} \delta \norm{\x}
	\end{equation*}
	and 
	\[
	\norm{ \J_{\mathcal{T}_{k+1}}\left(\x^k\right) } \le \sqrt{\left(4+\delta\right)} \norm{\x^k} \le \sqrt{\left(4+\delta\right)} (1+\delta) \norm{\x}.
	\]
	This gives
	\begin{equation*}
		I_2 =\norm{\J_{\mathcal{T}_{k+1}}(\x^k)^{\TT}\J_{\mcalT_k}(\h^k)\h^k_{\mcalT_k}}\le 2\zeta_3\norm{\h^k},
	\end{equation*}
	where $\zeta_3:=(2+\delta/2)(1+\delta)\delta\norm{\x}^2$. 
	
	For the third term $I_3$, it follows from \eqref{eq:sec} and (\ref{boundxk}) with $\p=\x^k$, $\mcalS = \mcalT_{k+1}$ and $\mcalT = \mcalT_{k}$ that
	\begin{equation*}
		\norm{\J_{\mathcal{T}_{k+1}}(\x^k)^{\TT}\J_{\mathcal{T}_k \backslash{\mathcal{T}_{k+1}}}(\x^k)} \le (1+\delta)\norm{\x^k}^2\le  (1+\delta)^3\norm{\x}^2,
	\end{equation*}
	which implies
	\begin{equation*}
		I_3=\norm{\J_{\mathcal{T}_{k+1}}(\x^k)^{\TT}\J_{\mathcal{T}_k \backslash \mathcal{T}_{k+1}}(\x^k)\h^k_{\mcalT_k\backslash \mcalT_{k+1}}} \le \zeta_4 \norm{\h^k_{\mcalT_k\backslash \mcalT_{k+1}}},
	\end{equation*}
	where $\zeta_4:=(1+\delta)^3\norm{\x}^2$.
	
	Finally, for the last term $I_4$, according to (\ref{boundxk}), we obtain
	\begin{equation*}
		I_4 \le C_0 \sigma (1+\delta)\norm{\x}  \sqrt{ \frac{s \log(en/s)}m}
	\end{equation*} 
	Combining $I_1$, $I_2$, $I_3$, $I_4$ together with \eqref{eq:2vxt},  we obtain
	\begin{eqnarray*}
		&&\norm{\u^k-\x} \\
		& \le &  \frac{\sqrt5+1}{2} \norm{\v^k_{\mathcal{T}_{k+1}}-\x_{\mathcal{T}_{k+1}}} \\
		&\le &    \frac{\sqrt5+1}{2} \cdot \Big( \max\{1-\mu^k\zeta_1, \mu^k\zeta_2-1 \} \| \h^k_{\mcalT_{k+1}}\| +\mu^k \zeta_3 \norm{\h^k} + \mu^k  \zeta_4 \norm{\h^k_{\mcalT_k\backslash \mcalT_{k+1}}} 
		\Big)  \\
		&&+ \frac{\sqrt5+1}{2} \cdot\mu^k \cdot C_0 \sigma (1+\delta)\norm{\x}  \sqrt{ \frac{s \log(en/s)}m}\\
		& \le &  \frac{\sqrt5+1}{2} \cdot \left( \sqrt2 \max\left\{ \max\{1-\mu^k\zeta_1, \mu^k\zeta_2-1 \}, \mu^k \zeta_4 \right\} + \mu^k \zeta_3 \right) \norm{\h^k} \\
		&&+ \frac{\sqrt5+1}{2} \cdot\mu^k \cdot C_0 \sigma (1+\delta)\norm{\x}  \sqrt{ \frac{s \log(en/s)}m}\\
		&:=& \rho \norm{\x^k-\x}+ \frac{\sqrt5+1}{2} \cdot\mu^k \cdot C_0 \sigma (1+\delta)\norm{\x}  \sqrt{ \frac{s \log(en/s)}m},
	\end{eqnarray*}
	where the third inequality follows from $\| \h^k_{\mcalT_{k+1}}\| + \| \h^k_{\mcalT_k\backslash \mcalT_{k+1}}\|  \le \sqrt 2  \norm{\h^k}$.  It is easy to check that when $0<\delta \le 0.01$,  there exist constants $\underline\mu, \overline\mu \in (0,1)$ such that when $\mu^k \in (\frac{\underline\mu}{\norm{\x}^2}, \frac{\overline\mu}{\norm{\x}^2})$, it holds $\rho\in (0,0.95)$. For instance, taking $\delta=0.01$, one obtains the feasible range for step size $\mu^k$ with
	\[
	\underline\mu =0.303, \quad \overline\mu=0.344.
	\]
	This completes the proof.
\end{proof}

With this in place, we are ready to prove Theorem \ref{Theorem3}.
\begin{proof}[Proof of Theorem \ref{Theorem3}]
	Due to the rotation invariance, without loss of generality we assume that $\norm{\x^0-\x} \le \norm{\x^0+\x}$. 
	Thus, the distance between the initial vector $\x^0$  and the target vector $\x$ is ${\rm dist}(\x^0,\x) = \norm{\x^0-\x}$. 
	Define $\mcalS_{k+1}=\mbox{supp}(\x^{k+1})$, $\mcalS^* = \mbox{supp}(\x)$ and $\mathcal{T}_k=\mbox{supp}(\x^k) \cup \mbox{supp}(\x)$. 
	To prove the convergence of the algorithm, observe that
	\begin{equation} \label{eq:upxkx}
		\begin{aligned}
			\norm{\x^{k+1}-\x}
			&=\left[\norm{\x^{k+1}_{\mcalS_{k+1}}-\x_{\mcalS_{k+1}}}^2+\norm{\x^{k+1}_{\mcalS^{c}_{k+1}}-\x_{\mcalS^c_{k+1}}}^2\right]^{\frac{1}{2}}\\ 
			&=\left[\norm{\x^{k+1}_{\mcalS_{k+1}}-\x_{\mcalS_{k+1}}}^2+\norm{\x_{\mcalS^{c}_{k+1}}}^2\right]^{\frac{1}{2}}.
		\end{aligned}
	\end{equation}
	For convenience, set
	\[
	\omega_1=\norm{\x^{k+1}_{\mcalS_{k+1}}-\x_{\mcalS_{k+1}}}, \qquad \omega_2=\norm{\x_{\mcalS^{c}_{k+1}}}.
	\]
	We then proceed to bound the above two terms. 
	Note that  $\textstyle\x_{\mcalS_{k+1}^{c}}$ is a subvector of $\x-\u^k$ due to the fact that $\mbox{supp}(\u^k)=\mcalS_{k+1}$. Here, $\u^k= \mathcal{H}_s(\x^k-\mu^k\nabla f(\x^k))$ is defined in Algorithm \ref{algorithm2:Iterative Algorithm}.  Under the condition of Theorem \ref{Theorem3}, it follows from Lemma \ref{lemma7} that  events \eqref{eq:concentration1} and \eqref{eq:concentration2} hold.
	Applying Lemma \ref{lemma9}, we obtain
	\begin{equation} \label{eq:ome2}
		\omega_2=\| \x_{\mcalS^{c}_{k+1}}\|  \le \norm{\x-\u^k} \le \rho \norm{\x^k-\x} + \frac{\sqrt5+1}{2} \cdot\mu^k \cdot C_0 \sigma (1+\delta)\norm{\x}  \sqrt{ \frac{s \log(en/s)}m},
	\end{equation}
	where $\rho \in (0,0.95)$. We next proceed to bound the term $\omega_1$.
	The update rule \eqref{eq:updarule} together with \eqref{eq:gradf} and \eqref{eq:defpk} gives
	\begin{eqnarray*}
		\omega_1 & = & \norm{\x^k_{\mcalS_{k+1}}-\x_{\mcalS_{k+1}}-\left[\J_{\mcalS_{k+1}}(\x^k)^{\TT}\J_{\mcalS_{k+1}}(\x^k)\right]^{-1}\left(\J_{\mcalS_{k+1}}(\x^k)^{\TT} \F(\x^k) \right.\right.\\
			&&\left.\left.
			-\J_{\mcalS_{k+1}}(\x^k)^{\TT}\J_{\mcalS_{k+1}^c}(\x^k)\x^k_{\mcalS^c_{k+1}}\right)}  \\
		&= & \norm{\left[\J_{\mcalS_{k+1}}(\x^k)^{\TT}\J_{\mcalS_{k+1}}(\x^k)\right]^{-1}\left( \J_{\mcalS_{k+1}}(\x^k)^{\TT}\J(\x^k)  \h^k - \J_{\mcalS_{k+1}}(\x^k)^{\TT} \F(\x^k) \right.\right.\\
			&& \left.\left. + \J_{\mcalS_{k+1}}(\x^k)^{\TT}\J_{\mcalS_{k+1}^c}(\x^k)\x_{\mcalS^c_{k+1}}\right) },
	\end{eqnarray*}
	where the second equality comes from the fact that 
	\[
	\J_{\mcalS_{k+1}}(\x^k)^{\TT}\J_{\mcalS_{k+1}}(\x^k) \h^k_{\mcalS_{k+1}}=\J_{\mcalS_{k+1}}(\x^k)^{\TT}\J(\x^k)  \h^k - \J_{\mcalS_{k+1}}(\x^k)^{\TT}\J_{\mcalS_{k+1}^c}(\x^k) \h^k_{\mcalS_{k+1}^c}.
	\]
	Here, $\h^k=\x^k-\x$.
	Note that $\mbox{supp}(\h^k) \subset \mathcal{T}_k$. 
	
	The fundamental theorem of calculus together with the fact $\F(\x)=0$ with $y_i=\x^T\A_i\x$, Newton-Leibniz formula $F(\x^k)-F(\x) = \int_0^1 \nabla F(\x(t))(\x^k-\x) dt$ and Lemma \ref{lemma8} gives
	\begin{equation}
		\begin{aligned}
			\label{eq:wme2} 
			\omega_1 &\le  \frac{1}{(2-\delta)\norm{\x^k}^2}\cdot \left\| \int_0^1 \J_{\mcalS_{k+1}}(\x^k)^{\TT} \big( \J_{\mathcal{T}_k}(\x^k) -\J_{\mathcal{T}_k}(\x(t)) \big)\h_{\mathcal{T}_k}^k dt \right.\\
			&\left.+ \J_{\mcalS_{k+1}}(\x^k)^{\TT}\J_{\mcalS_{k+1}^c}(\x^k)\x_{\mcalS^c_{k+1}}
			+\frac{\mu^k}{m}\sum\limits_{i=1}^{m}\epsilon_i[\widetilde{\A}_i]_{\mcalS_{k+1},:}\x^k \right\|\\
			&\le  \frac{1}{(2-\delta)\norm{\x^k}^2}\cdot \left( \int_0^1 \norm{\J_{\mcalS_{k+1}}(\x^k)  } \norm{\J_{\mathcal{T}_k}(\x^k) -\J_{\mathcal{T}_k}(\x(t))} \norm{\h^k} dt \right. \\
			& \left.+ \norm{\J_{\mcalS_{k+1}}(\x^k)^{\TT}\J_{\mcalS_{k+1}^c}(\x^k) }  \norm{\x_{\mcalS^c_{k+1}}  }+\norm{\frac{\mu^k}{m}\sum\limits_{i=1}^{m}\epsilon_i[\widetilde{\A}_i]_{\mcalS_{k+1},:}\x^k}\right)  \\
			&\le  \frac{1}{(2-\delta)\norm{\x^k}^2} \cdot \left( (4+\delta)\norm{\x^k} \norm{\h^k}^2 \int_0^1 (1-t)dt + (1+\delta)\norm{\x^k}^2 \norm{\x_{\mcalS^c_{k+1}} }\right.\\
			& \left.+\norm{ \frac{\mu^k}{m}\sum\limits_{i=1}^{m} \epsilon_i [\widetilde{\A}_i]_{\mcalS_{k+1},:} \x^k} \right)  \\
			&\le  \frac{4+\delta}{2(2-\delta)(1-\delta)\norm{\x}}  \norm{\x^k-\x}^2 +\frac{(1+\delta)\rho}{2-\delta} \norm{\x^k-\x}+\frac{C(\delta)\mu^k \sigma}{\norm{\x}}  \sqrt{ \frac{s \log(en/s)}m}, 
		\end{aligned}
	\end{equation}
	where we denote $\x(t)=\x+t(\x^k-\x)$. Here, the third inequality comes from Lemma \ref{lemma8},   the linearity of $\J(\x^k)$, and the fact that 
	\[
	\x^k-\x(t)=(1-t)(\x^k-\x),
	\]
	and the last inequality arises from \eqref{eq:ome2}, Lemma \ref{le:epAi} with $\mcalS:=\mathcal{S}_{k+1} \cup \mbox{supp}(\x^k)$ and the condition $\norm{\x^k-\x} \le \delta \norm{\x}$. Here, $C(\delta)>0$ is a constant depending on $\delta$.
	Putting everything together, we obtain 
	\begin{equation*}
		\norm{\x^{k+1}-\x}  \le \rho' \norm{\x^k-\x}+\frac{C \sigma}{\norm{\x}}  \sqrt{ \frac{s \log(en/s)}m},
	\end{equation*}
	where  $\rho' \in (0,1)$ by taking $\delta$ to be sufficiently small and $\mu^k \in (\frac{\underline\mu}{\norm{\x}^2}, \frac{\overline\mu}{\norm{\x}^2})$ for some universal constants $\underline\mu, \overline\mu \in (0,1)$. This completes the proof of the first part of Theorem \ref{Theorem3}.

	For the rest of Theorem \ref{Theorem3}, it is easy to see that when $\sigma = 0$, the aforementioned derivation remains valid. Note that \eqref{eq:upxkx} remains valid in noise-free case, namely, 
	\begin{equation}
		\norm{\x^{k+1}-\x}=\left[\norm{\x^{k+1}_{\mcalS_{k+1}}-\x_{\mcalS_{k+1}}}^2+\norm{\x_{\mcalS^c_{k+1}}}^2\right]^{\frac{1}{2}}.
	\end{equation}
	Similar to \eqref{eq:ome2}, the second term can be upper bounded by $\norm{\u^k-\x}$,
	\begin{equation*}
		\omega_2=\| \x_{\mcalS^{c}_{k+1}}\|  \le \norm{\x-\u^k} \le \rho \norm{\x^k-\x}.
	\end{equation*}
	For the first term, using the same approach as in \eqref{eq:wme2}, we have
	\begin{eqnarray}
		\omega_1 
		&\le  &\frac{1}{(2-\delta)\norm{\x^k}^2}\cdot \left( \int_0^1 \norm{\J_{\mcalS_{k+1}}(\x^k)  } \norm{\J_{\mathcal{T}_k}(\x^k) -\J_{\mathcal{T}_k}(\x(t))} \norm{\h^k} dt \right.\notag \\
		&& \left.+ \norm{\J_{\mcalS_{k+1}}(\x^k)^{\TT}\J_{\mcalS_{k+1}^c}(\x^k) }  \norm{\x_{\mcalS^c_{k+1}}  }\right)  \label{eq: ome_1'} \\
		&\le&  \frac{4+\delta}{2(2-\delta)(1-\delta)\norm{\x}}  \norm{\x^k-\x}^2 +\frac{(1+\delta)\rho}{2-\delta} \norm{\x^k-\x}, \notag
	\end{eqnarray}
	Putting the estimators $\omega_1$ and $\omega_2$ into \eqref{eq:upxkx}, we obtain
	\begin{equation} \label{eq:upmu}
		\norm{\x^{k+1}-\x} \le  \nu \norm{\x^k-\x}
	\end{equation}
	for some $\nu \in (\rho,1)$.
	This completes the proof of  the first half of the remaining part of Theorem \ref{Theorem3}. 
	
	Next, we proceed to finish the last part of the theorem's proof.
	According to \eqref{eq:upmu}, it holds
	\[
	\norm{\x^k-\x} \le \nu^k \norm{\x^0-\x} \le \nu^k \delta \norm{\x}.
	\]
	Let $k_1>0$ be the smallest integer such that it holds
	\begin{equation}
		\nu^k \delta  \norm{\x} \le x_{\min}.
		\label{eq: quadratic}
	\end{equation}
	We claim that $\mcalS^* \subseteq \mcalS_k$ for all $k\ge k_1$, where $\mcalS_k =\mbox{supp}(\x^k)$. Indeed, if there exists a $k_0>k_1$ such that $\mcalS^* \not\subseteq \mcalS_{k_0} $, then there would exist $ i \in \mcalS^*\backslash \mcalS_{k_0} \neq \emptyset$. This implies $\textstyle \norm{\x^{k_0}-\x} \ge |x_{i}| \ge x_{\min}$, which contradicts \eqref{eq: quadratic} that  $\textstyle \norm{\x^{k_0}-\x} \le \nu^{k_0}  \delta  \norm{\x} <  \nu^{k_1}  \delta  \norm{\x} \le  x_{\min}$. This gives the claim. 
	With this in place, for all $k\ge k_1$,   we have $\x_{\mcalS^c_{k+1}}=0$. It then follows from \eqref{eq:upxkx} and \eqref{eq: ome_1'} that 
	\[
	\norm{\x^{k+1}-\x}= \norm{\x^{k+1}_{\mcalS_{k+1}}-\x_{\mcalS_{k+1}}} \le  \frac{4+\delta}{2(2-\delta)(1-\delta)\norm{\x}}  \norm{\x^k-\x}^2 = \nu' \norm{\x^k-\x}^2 
	\]
	for all $k\ge k_1$, where $\nu':=\frac{4+\delta}{2(2-\delta)(1-\delta)\norm{\x}} $.
	Finally, by \eqref{eq: quadratic}, we observe 
	\[
	k_1 = \left\lfloor\frac{\log{(\delta \norm{\x}/x_{\min})}}{\log{(\nu^{-1})}} \right\rfloor + 1 \le c_1\log{( \norm{\x}/x_{\min})}+c_2,
	\]
	we complete the proof.
\end{proof}

\bibliographystyle{amsplain}

\end{document}